\documentclass{article}

\usepackage[ruled,linesnumbered]{algorithm2e}

\usepackage[accepted]{icml2024}

\usepackage{amsmath,amsfonts,bm}

\def\sref#1{\S~\ref{#1}}

\def\eqref#1{equation~\ref{#1}}
\def\Eqref#1{Eq.(\ref{#1})}

\def\1{\bm{1}}

\def\va{{\bm{a}}}

\def\vs{{\bm{s}}}
\def\vt{{\bm{t}}}

\def\vw{{\bm{w}}}
\def\vx{{\bm{x}}}

\def\mX{{\bm{X}}}

\DeclareMathAlphabet{\mathsfit}{\encodingdefault}{\sfdefault}{m}{sl}
\SetMathAlphabet{\mathsfit}{bold}{\encodingdefault}{\sfdefault}{bx}{n}

\newcommand{\E}{\mathbb{E}}

\usepackage[utf8]{inputenc} 
\usepackage[T1]{fontenc}    
\usepackage{hyperref}       
\usepackage{url}            
\usepackage{booktabs}       
\usepackage{amsfonts}       
\usepackage{nicefrac}       
\usepackage{microtype}      
\usepackage{xcolor}         
\usepackage{amsmath,amssymb,amsfonts,amsthm,mathrsfs}
\usepackage[page,title,titletoc,header]{appendix}
\usepackage{enumerate}
\usepackage{color}
\usepackage{graphicx}
\usepackage{indentfirst}
\usepackage[mathscr]{eucal}
\usepackage{multirow}
\usepackage{algpseudocode}
\newtheorem{theorem}{Theorem}
\newtheorem{lemma}{Lemma}[section]
\newtheorem{proposition}{Proposition}

\newtheorem{definition}{Definition}
\newtheorem{remark}{Remark}
\newtheorem{example}{Example}
\usepackage{nicefrac}
\usepackage{subfigure}
\usepackage{enumitem}
\usepackage{graphicx}
\usepackage{wrapfig}

\usepackage{CJKutf8}
\usepackage{bbm}
\usepackage{colortbl}

\newcommand{\eg}{{\it e.g.}}

\newcommand{\ie}{{\it i.e.}}

\icmltitlerunning{Identifiability Matters: Revealing the Hidden Recoverable Condition in Unbiased Learning to Rank}

\begin{document}

\twocolumn[
\icmltitle{Identifiability Matters: Revealing the Hidden Recoverable Condition in Unbiased Learning to Rank}

\icmlsetsymbol{equal}{*}

\begin{icmlauthorlist}
    \icmlauthor{Mouxiang Chen}{zju}
    \icmlauthor{Chenghao Liu}{salesforce}
    \icmlauthor{Zemin Liu}{nus}
    \icmlauthor{Zhuo Li}{statestreet}
    \icmlauthor{Jianling Sun}{zju}
\end{icmlauthorlist}

\icmlaffiliation{zju}{Zhejiang University}
\icmlaffiliation{salesforce}{Salesforce Research Asia}
\icmlaffiliation{nus}{National University of Singapore}
\icmlaffiliation{statestreet}{State Street Technology (Zhejiang) Ltd}

\icmlcorrespondingauthor{Chenghao Liu}{chenghao.liu@salesforce.com}
\icmlcorrespondingauthor{Zhuo Li}{lizhuo@zju.edu.cn}

\icmlkeywords{learning to rank, unbiased learning to rank, identifiability}

\vskip 0.3in
]



\printAffiliationsAndNotice{}  

\begin{abstract}

Unbiased Learning to Rank (ULTR) aims to train unbiased ranking models from biased click logs, by explicitly modeling a generation process for user behavior and fitting click data based on examination hypothesis. Previous research found empirically that the true latent relevance is mostly recoverable through click fitting. However, we demonstrate that this is not always achievable, resulting in a significant reduction in ranking performance. This research investigates the conditions under which relevance can be recovered from click data in the first principle. We initially characterize a ranking model as \textit{identifiable} if it can recover the true relevance up to a scaling transformation, a criterion sufficient for the pairwise ranking objective. Subsequently, we investigate an equivalent condition for identifiability, articulated as a graph connectivity test problem: the recovery of relevance is feasible if and only if the \textit{identifiability graph} (IG), derived from the underlying structure of the dataset, is connected. The presence of a disconnected IG may lead to degenerate cases and suboptimal ranking performance. To tackle this challenge, we introduce two methods, namely \textit{node intervention} and \textit{node merging}, designed to modify the dataset and restore the connectivity of the IG. Empirical results derived from a simulated dataset and two real-world LTR benchmark datasets not only validate our proposed theory but also demonstrate the effectiveness of our methods in alleviating data bias when the relevance model is unidentifiable.

\end{abstract}

\section{Introduction}

The utilization of click data for Learning to Rank (LTR) methodologies has become prevalent in contemporary information retrieval systems. The accumulated feedback effectively demonstrates the value of individual documents to users \citep{agarwal2019general} and is comparatively effortless to amass on an extensive scale. Nevertheless, inherent biases stemming from user behaviors are presented within these datasets \citep{joachims2007evaluating}. One example is position bias \citep{joachims2005accurately}, wherein users exhibit a propensity to examine documents situated higher in the rankings, causing clicks to be biased with the position. Removing these biases is critical as they impede learning the correct ranking model from click data. To address this issue, Unbiased Learning to Rank (ULTR) is developed to mitigate these biases \citep{joachims2017unbiased}. The central idea is to explicitly model a generation process for user clicks using the \textbf{examination hypothesis}. This hypothesis posits that each document has a probability of being observed (depending on certain bias factors, \eg, position \citep{joachims2017unbiased} or context \citep{fang2019intervention}), and subsequently clicked based on its relevance to the query (depending on ranking features encoding the query and document). It can be formulated as:
\begin{align*}
    P(\text{click})=\underbrace{P(O\mid \text{bias factors})}_{\text{observation model}}\cdot \underbrace{P(R\mid \text{ranking features})}_{\text{ranking model}}.
\end{align*}
In practice, a joint optimization is utilized to optimize both the observation and ranking models based on the examination hypothesis \citep{wang2018position,ai2018unbiased,zhao2019recommending,guo2019pal}. This approach efficiently accommodates the fitting of click data.

However, our focus extends beyond predicting the click probability; we are more interested in recovering the true relevance through the ranking model, which is the primary objective of ULTR. Previous empirical studies have suggested that this objective can be achieved in most cases after the click probability is fitted \citep{ai2018unbiased,wang2018position,ai2021unbiased}. Regrettably, no existing literature provides a theoretical guarantee to substantiate this claim, and it is possible that the relevance becomes \textit{unrecoverable} in certain cases. This was demonstrated in \citet{oosterhuis2022reaching}, which constructed a simplified dataset, revealing that a perfectly trained click model on this dataset can still produce inaccurate and inconsistent relevance estimates. Despite the limited scale of their example, we argue that in real scenarios with large-scale data, the unrecoverable phenomenon may persist, particularly in the presence of excessive bias factors. We illustrate it in the following example.



\begin{example}
\label{exm:wrong_ranking}
Consider a large-scale dataset where each displayed query-document pair comes with a \underline{\textbf{distinct}} bias factor\footnote{Achieving this is possible by incorporating sufficiently fine-grained bias factors identified in previous studies, such as positions, other document clicks \citep{vardasbi2020cascade,chen2021adapting}, contextual information \citep{fang2019intervention}, representation styles \citep{liu2015influence,zheng2019constructing}, or various other contextual features \citep{ieong2012domain,sun2020eliminating,sarvi2023impact}.}. Unfortunately, in such cases, it becomes challenging to discern whether the influence on relevance arises from the ranking feature or the bias factor itself \citep{chen2022lbd,zhu2020unbiased}. One might train a naive ranker that predicts a constant one, coupled with an observation model that maps the unique bias factor to the true click probability. While the product of the two models yields an accurate click estimation, the estimated relevance is erroneous.
\end{example}

Given the recent research emphasis on enhancing datasets by incorporating additional bias factors to refine observation estimation \citep{vardasbi2020cascade,chen2021adapting,chen2022lbd,chen2023multi}, we contend that there exists an urgent need to address a fundamental question: \emph{under what circumstances can the relevance be recovered?} However, tackling this problem is notably challenging as it heavily relies on the exact data collection procedure \citep{oosterhuis2022reaching}. To the best of our knowledge, existing research has not yet presented the condition of relevance recovery at a fundamental level.

In this work, we present a general identifiability framework to address the core of relevance recovery. Importantly, this framework starts from the most basic examination hypothesis and does not impose additional requirements on the model implementation, optimization process, or specific types of bias. Focusing on the ranking objective, we define a ranking model as \textbf{\emph{identifiable}} in \sref{sec:identifiability} if it can recover the true relevance probability \emph{up to a scaling transformation}. The absence of identifiability can lead to degenerate scenarios and suboptimal ranking performance. We establish that the identifiability of a ranking model is linked to the underlying structure of the dataset, which can be translated into a practical graph connectivity test problem based on an \emph{identifiability graph} (IG) related to the dataset. Specifically, if and only if the IG is connected, we can ensure the identifiability of the ranking model. Disconnections in the IG may result in inaccurate rankings. Theoretical and empirical analyses unveil that the identifiability probability is influenced by the dataset's size, bias factors, and features. Smaller datasets, more features, or an increased number of bias factors elevate the likelihood of IG disconnection.
Furthermore, our analysis of the TianGong-ST dataset \citep{chen2019tian} indicates that the unidentifiability issue is present in the real world. These insights provide valuable guidance for the future design of ULTR datasets and algorithms.

Building upon the aforementioned theory, we delve into strategies for addressing datasets characterized by a disconnected IG. We introduce two methods aimed at enhancing IG connectivity to ensure the identifiability of ranking models: (1) \textit{node intervention} (\sref{sec:node_intervention}), which swaps documents between two bias factors to augment the dataset. While the intervention is common in ULTR \citep{joachims2017unbiased}, our approach explores the minimal required number of interventions and significantly reduces the online cost; and (2) \textit{node merging} (\sref{sec:node_merging}), which merges two bias factors and assumes identical observation probabilities. We conducted extensive experiments using a fully simulated dataset and two real-world LTR datasets to validate our theorems and demonstrate the efficacy of our methods when dealing with disconnected IGs. 

To the best of our knowledge, we are the first to study the identifiability of ranking models for ULTR in the first principle. Our main contributions are summarized as follows:
\begin{enumerate}[leftmargin=*]
    \item We propose the concept of identifiability, ensuring the capacity to recover the true relevance from click data, and frame it as a graph connectivity test problem from the dataset perspective.
    \item We propose model-agnostic methods to handle the unidentifiable cases by restoring graph connectivity.
    \item We provide both theoretical guarantees and empirical studies to verify our proposed framework.
\end{enumerate}
\section{Preliminaries}

Given a query $q\in \mathcal{Q}$, the goal of learning to rank is to learn a ranking model to sort a set of documents. Let $\vx\in\mathcal{X}$ denote the query-document ranking features, and the ranking model aims to estimate a relevance score with $\vx$. In practical scenarios where acquiring ground truth relevance at scale is challenging, users' click logs are frequently employed as labels for training the model \cite{joachims2005accurately}. While user clicks usually exhibit bias compared to the true relevance, researchers propose examination hypothesis to decompose the biased clicks into relevance probability and observation probability, formulated as:
\begin{align}
    \label{eq:examination_hypothesis}
    c(\vx, \vt)=r(\vx)\cdot o(\vt),
\end{align}
where $c(\vx, \vt)$, $r(\vx)$, and $o(\vt)$ denote the probabilities that the document is clicked, relevant, and observed, respectively. In this paper, we focus on the scenario in which the true observation probability $o$ and relevance probability $r$ are both \textit{unknown}. $\vt\in\mathcal{T}$ denotes bias factors that introduce clicks to be biased, such as position \citep{joachims2017unbiased}, other document clicks \citep{vardasbi2020cascade,chen2021adapting}, contextual information \citep{fang2019intervention} or the representation style \citep{liu2015influence}. Let $\mathcal{D}=\{(\vx_i, \vt_i)\}_{i=1}^{|\mathcal{D}|}$ denote pairs of ranking features and bias factors. To simplify the analysis, we suppose all bias factors $\vt\in\mathcal T$ and features $\vx\in\mathcal X$ appear in $\mathcal D$. By explicitly modeling the bias via observation, we can attain an unbiased estimate of the ranking objective.

To jointly obtain the relevance score and observation score, we optimize a ranking model $r'(\cdot)$ and an observation model $o'(\cdot)$ to fit clicks in the form of the examination hypothesis. For example, two-tower \cite{guo2019pal} used the following objective:
\begin{align}
    \label{eq:loss}
    \mathcal L = \sum_{i=1}^{|\mathcal D|} l(r'(\vx_i)\cdot o'(\vt_i), c_i),
\end{align}
where $c_i$ denotes the click, and $l(\cdot,\cdot)$ denotes a loss function of interest, such as mean square error or cross-entropy error.
\section{Identifiability}

\label{sec:identifiability}

Most current work presumes that optimizing the two models in the form of \Eqref{eq:examination_hypothesis} can yield an accurate ranking model \citep{wang2018position,chen2022lbd,chen2023multi}.
Regrettably, there are no guarantees regarding this assumption. While it is established that the product of the outputs from the two models is accurate, there are instances, as demonstrated in Example \ref{exm:wrong_ranking}, where this approach results in poor ranking performance. Our objective is to investigate a fundamental condition under which the underlying latent relevance function can be recovered (\textit{i.e.,} \textbf{identifiable}), formulated as:
\begin{align*}
    r(\vx)\cdot o(\vt)=r'(\vx)\cdot o'(\vt) \quad \Longrightarrow \quad r = r'.
\end{align*}

In this paper, we intentionally prevent imposing specific constraints on the models of $r'$ and $o'$, making our approach broadly adaptable to existing work based on the examination hypothesis. It is worth noting that directly recovering an exact relevance model is impractical, as scaling $r(\cdot)$ by a factor of $n$ and $o(\cdot)$ by $\nicefrac{1}{n}$ would leave their product unchanged. In practice, we are often interested in making relevance identifiable up to a scaling transformation, which is sufficient for pairwise ranking objectives. Consequently, we introduce the following definition of identifiability:

\begin{definition}[Identifiability]
\label{def:identifiable}
We say that the relevance model is identifiable, if:
\begin{align*}
    r(\vx)\cdot o(&\vt)=r'(\vx)\cdot o'(\vt), \quad\forall (\vx, \vt)\in \mathcal{D} \\
    &\Longrightarrow \quad  \exists C>0, \text{ s.t. }r(\vx)=C r'(\vx), \forall  \vx\in \mathcal X.
\end{align*}
\end{definition}

Identifiability serves as a sufficient condition for ensuring accurate ranking. The absence of such guarantees can lead to degenerate scenarios, as illustrated in Example \ref{exm:wrong_ranking}. Our following main result (with proof in Appendix \ref{sec:app_proof_1}) establishes that identifiability is intrinsically linked to the dataset's underlying structure, which can be readily mined.

\begin{theorem}[\textbf{Main result:} Equivalent condition of identifiability]
\label{thm:graph}
The relevance model is identifiable, if and only if an undirected graph $G=(V,E)$ is connected, where $V$ is a node set and $E$ is an edge set, defined as:
\begin{align*}
    V=&\{v_1, v_2, \cdots, v_{|\mathcal{T}|}\},\\
    E=&\left\{(v_s,v_t)\mid \exists \vx\in \mathcal X \text{, s.t. }(\vx,\vs)\in \mathcal D \land (\vx,\vt)\in \mathcal D\right\},
\end{align*}
We refer to this graph as \textbf{identifiability graph} (IG).
\end{theorem}
\begin{remark}
The relevance identifiability is equivalent to a graph connectivity test problem. The IG is constructed as follows: we first create nodes for each bias factor $t\in\mathcal T$. If there exists a feature appearing with two bias factors together, add an edge between the two nodes. Theorem \ref{thm:graph} establishes connections to recent ULTR research \citep{agarwal2019estimating,oosterhuis2022reaching,zhang2023towards}, which are elaborated in \sref{sec:related-work}.
\end{remark}

\begin{figure}[t]
    \centering
    \includegraphics[width=0.48\textwidth]{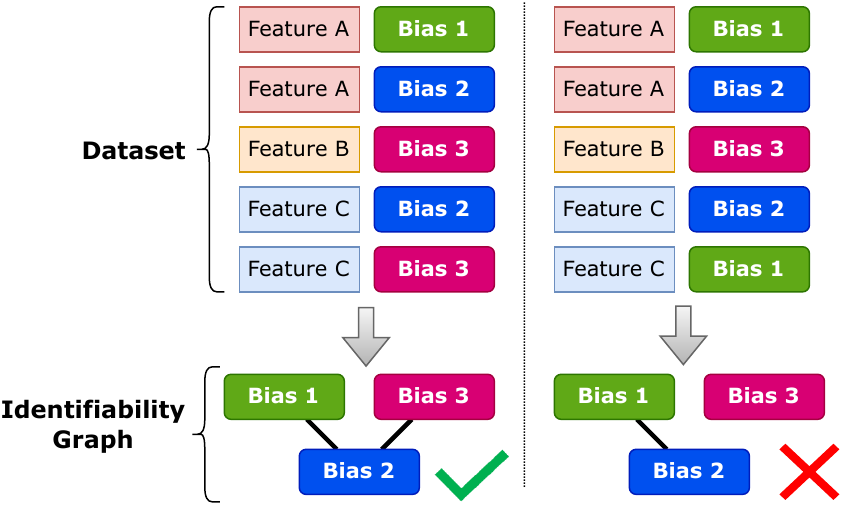}
    \caption{Examples for identifiable case and unidentifiable case.}
    \label{fig:identifiable_example}
\end{figure}

Figure \ref{fig:identifiable_example} illustrates examples for applying Theorem \ref{thm:graph} to verify the identifiability of the datasets. In the left figure, bias factors 1 and 2 are connected through feature 1, and bias factors 2 and 3 are connected through feature 3. As a result, the graph is connected and the relevance is identifiable. Conversely, the right figure depicts a scenario where bias factor 3 remains isolated, leading to unidentifiability of relevance. Based on Theorem \ref{thm:graph}, we illustrate the identifiability check algorithm in Appendix \ref{sec:app_algo_idcheck}.

\begin{figure*}[htbp]
    \centering
    \subfigure[node intervention]{
        \includegraphics[width=0.48\textwidth]{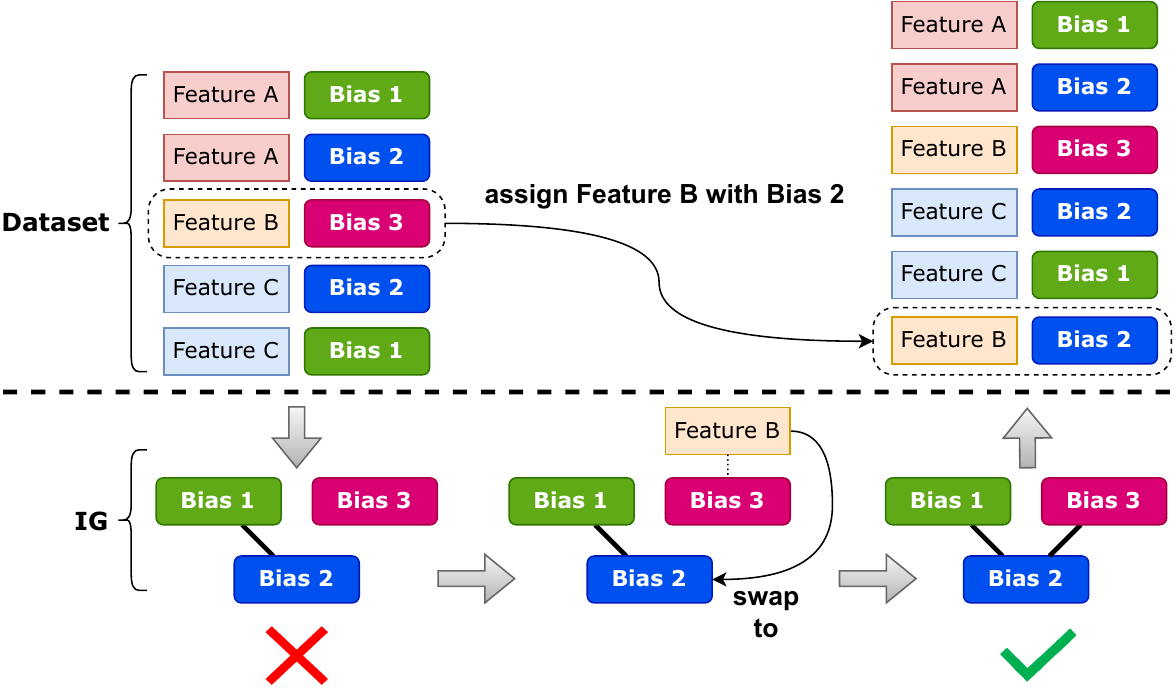}
        \label{fig:method_ni}
    }
    \subfigure[node merging]{
        \includegraphics[width=0.48\textwidth]{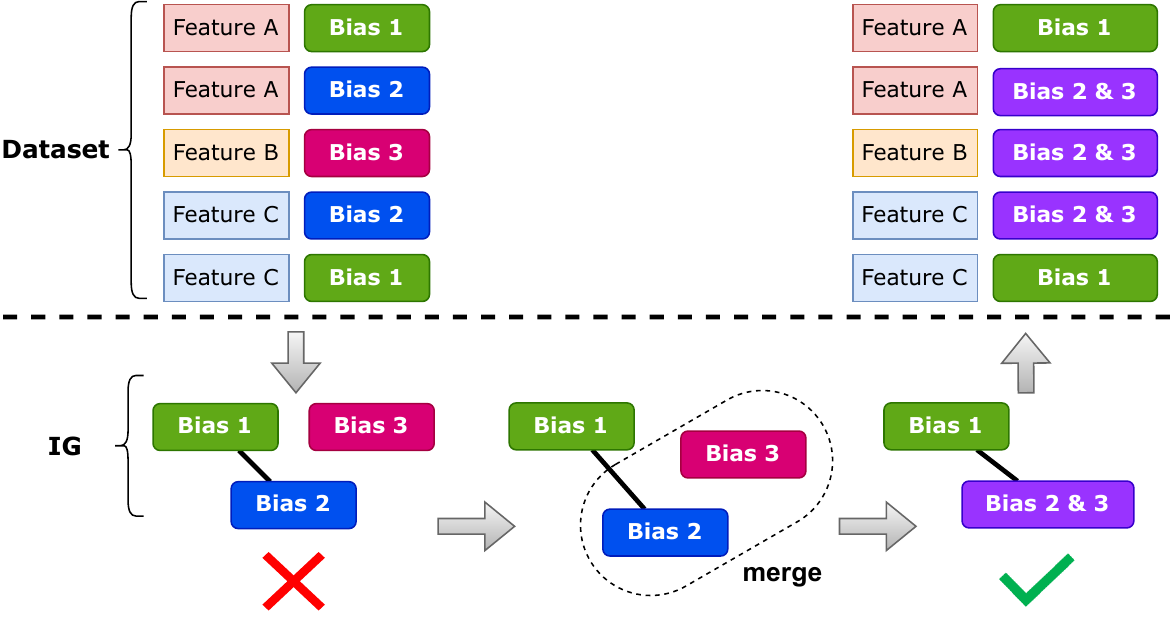}
        \label{fig:method_nm}
    }
    \caption{Illustrations for the proposed two methods to deal with unidentifiable datasets. In node intervention, we swap Feature B related to Bias 3 to Bias 2, which connects the two disconnected nodes. In node merging, we merge Bias 2 and Bias 3 into a new node 2 \& 3, which indicates that Bias 2 and Bias 3 will have the same estimated observation. Both methods are applied to datasets before ULTR training.}
    \label{fig:illustration_method}
\end{figure*}

Our next objective is to ascertain the probability of a ranking model being identifiable, particularly focusing on how it scales with the sizes $|\mathcal{D}|$, $|\mathcal{X}|$, and $|\mathcal{T}|$, to offer an intuitive understanding. However, accurately calculating this probability is intricate due to the uncertain generation process of $\mathcal{D}$. To address this, we consider a simplified distribution for $\mathcal{D}$ in the following example (proof in Appendix \ref{sec:app_proof_2}). 

\begin{example}[Estimation of identifiability probability]
\label{thm:identifiability_prob}
Considering the following simplified example: each feature $\vx\in \mathcal X$ and bias factor $\vt\in \mathcal T$ are selected independently and uniformly to construct a dataset $\mathcal D$. Then the probability of identifiability can be estimated by:
\begin{align*}
P(\text{identifiability})&  \sim 1 - |\mathcal T| \exp\left(
    -|\mathcal D| + f
\right),\\
\text{where } f &= |\mathcal X||\mathcal {T}|\log \left[
    2 - \exp\left(
        -\frac{|\mathcal D|}{|\mathcal X||\mathcal {T}|}
    \right)
\right].
\end{align*}


\end{example}
\vspace{-3mm}
\begin{remark}
Note that $\lim_{|\mathcal {T}| \rightarrow +\infty} f = |\mathcal D|$. This implies a conclusion that when $|\mathcal T|$ is sufficiently large, the probability of identifiability decays, which requires a sufficiently large dataset $|\mathcal D|$ to offset this effect. While it is derived from a simplified case, this conclusion is consistent with our empirical evaluations in a more realistic setting (\sref{sec:experiment_large}). Given the recent research focused on enhancing datasets with extra bias factors while keeping dataset size constant \citep{chen2021adapting,chen2023multi,sarvi2023impact}, these insights provide potential valuable guidance for future work.
\end{remark}
\section{Dealing with unidentifiable dataset}

In this section, we discuss how to deal with datasets whose relevance is unidentifiable. We propose two methods applied on datasets, namely \textit{node intervention} and \textit{node merging}, to establish connectivity within the IG. The former method necessitates the inclusion of supplementary data which enables the accurate recovery of relevance, while the latter method only needs existing data but may introduce approximation errors. In practice, the choice depends on the specific requirements of the problem. 

\subsection{Node intervention}

\label{sec:node_intervention}

Given that unidentifiability results from the incomplete dataset, one method is to augment datasets by swapping some pairs of documents between different bias factors (mainly positions). Swapping is a widely adopted technique that can yield accurate observation probabilities \citep{joachims2017unbiased}. However, existing methods are often heuristic in nature, lacking a systematic theory to guide the swapping process, and sometimes performing unnecessary swapping operations. Instead, we leverage IGs to identify a minimal set of critical documents for effective swapping.

Our basic idea is that (1) find the connected components in the IG; (2) for any two components, find a node in each of them; (3) for these two nodes (representing two bias factors, \eg, $\vt_1$ and $\vt_2$), find a feature (denoted by $\vx$) related to one of them (denoted by $\vt_1$), and swap it to the other (denoted by $\vt_2$). This creates a new data point $(\vx, \vt_2)$. Since $(\vx, \vt_1)$ exists, this action bridges $\vt_1$ and $\vt_2$ in the IG and thus connects two components. Repeat this process until the IG is connected. We refer to this method as \textbf{node intervention}, which is illustrated in Figure \ref{fig:method_ni}.

It's worth mentioning that there are many choices to select the features and bias factors. To be specific, if we assume there are infinity data and the click probability $r(x)o(t)$ can be observed accurately, any intervention target that makes the IG connected can lead to identifiability, with no theoretical distinction in quality. However, note that we can only observe $N$ samples of $r(x)o(t)$ in practice, which causes some choices to be less effective. For example, the observation probability of some bias factors is relatively low (\eg, the last position in the list), necessitating more clicks to obtain a valid click rate estimation. These bias factors are not suitable for swapping. 

Therefore, to choose the optimal intervention target, instead of simply assuming that we observe an accurate click probability $r(\vx)\cdot o(\vt)$, here we assume that we can only observe a random variable for click rate which is an average of $N$ clicks: $\nicefrac{1}{N}\sum_{i=1}^{N}c_i$. $c_i\in\{0, 1\}$ is a binary random variable sampling from a probability $r(\vx)\cdot o(\vt)$, indicating a click occurrence. Definition \ref{def:identifiable} can be seen as a special case when $N\rightarrow + \infty$.
Based on it, we establish the following proposition, with its proof delegated to Appendix \ref{sec:app_proof_3}.

\begin{proposition}
\label{thm:variance}
For a feature $\vx$ and two bias factors $\vt_1, \vt_2$, suppose $r'(\vx)\cdot o'(\vt) =\nicefrac{1}{N}\sum_{i=1}^{N}c_i(\vx, \vt)$, where $\vt\in \{\vt_1, \vt_2\}$ and $c_i(\vx, \vt)$ are random variables i.i.d. sampled from $\text{Bernoulli}(r(\vx)\cdot o(\vt))$ for $1\leq i \leq N$. Assuming $r(\vx)$, $o(\vt_1)$ and $o(\vt_2)$ are non-zero, then:
\begin{align*}
    \E\left[\Omega \mid r'(\vx)
    \right]&=0,\\
    \mathbb{V}\left[
        \Omega\mid r'(\vx)
    \right]&=\frac{1}{NR}\left[
	\frac{1}{r(\vx)\cdot o(\vt_1)}+\frac{1}{r(\vx)\cdot o(\vt_2)}-2
\right],
\end{align*}
where $\Omega={\frac{o'(\vt_1)}{o(\vt_1)}-\frac{o'(\vt_2)}{o(\vt_2)}}$ and $R = \nicefrac{r'(\vx)^2}{r(\vx)^2}$. 
\end{proposition}
\begin{remark}
    The event $\Omega$ is closely related to the identifiability: As $N$ or $r(\vx)o(\vt)$ increases, the variance $\mathbb{V}$ of $\Omega$ decreases, leading to $\Omega\rightarrow0$ and $\nicefrac{o'(\vt)}{o(\vt)}$ approaching a constant, which indicates an identifiable ranking model according to Definition \ref{def:identifiable}. In practice, optimal feature and bias factors can be chosen to minimize the variance and reduce the required $N$.
\end{remark}

Based on Proposition \ref{thm:variance}, for two components $G_A=(V_A, E_A)$ and $G_B=(V_B, E_B)$, we use the following process to connect $G_A$ and $G_B$, by minimizing $\mathbb{V}$ to decreasing the necessary clicks and facilitate identifiability:
{\small
\begin{align}
    \label{eq:node_intervention_cost}
    &\mathcal C(\vx, \vt_1, \vt_2) = 
        \frac{1}{r(\vx)\cdot o(\vt_1)}+\frac{1}{r(\vx)\cdot o(\vt_2)}-2,\\
    \label{eq:node_intervention_21}
    &\vt_A^{(A)}, \vt_B^{(A)}, \vx_A = \underset{\vt_A\in V_A, \vt_B\in V_B, \vx\in X_{t_A}}{\arg\min} \mathcal C(\vx, \vt_A, \vt_B),\\
    \label{eq:node_intervention_12}
    &\vt_A^{(B)}, \vt_B^{(B)}, \vx_B = \underset{{\vt_A\in V_A, \vt_B\in V_B, \vx\in X_{t_B}}}{\arg\min}\mathcal C(\vx, \vt_A, \vt_B),\\
    \label{eq:node_intervention_d}
    &\vx^*, \vt^*=\begin{cases}
        \vx_A, \vt_B^{(A)} & \mbox{if $\mathcal C(\vx_A, \vt_A^{(A)}, \vt_B^{(A)})\leq \mathcal C(\vx_B, \vt_A^{(B)}, \vt_B^{(B)})$,}\\
        \vx_B, \vt_A^{(B)} & \mbox{otherwise,}\\
    \end{cases}
\end{align}
}
where $X_{t_A}=\{\vx_i\mid (\vx_i, \vt_A)\in\mathcal D\}$ and $X_{t_B}=\{\vx_i\mid (\vx_i, \vt_B)\in\mathcal D\}$. Here \Eqref{eq:node_intervention_cost} defines a cost\footnote{The value of $r$ and $o$ in \Eqref{eq:node_intervention_cost} can be based on rational estimations. For instance, $r$ can be chosen from a ranking model trained on biased clicks, and a manually crafted model can serve as $o$. Additionally, the node merging method, which we will soon introduce, may be utilized to generate initial estimates for $r$ and $o$.} of swapping $\vx$ from $\vt_1$ to $\vt_2$ (or from $\vt_2$ to $\vt_1$) based on $\mathbb V$ derived by Proposition \ref{thm:variance}. We ignore $R$ since it is a constant when the relevance model is identifiable. \Eqref{eq:node_intervention_21} defines the process that we find a feature $\vx_A$ related to a bias factor $\vt_A^{(A)}$ (belongs to $G_A$) and swap it to another bias factor $\vt_B^{(A)}$ (belongs to $G_B$). Reversely, \Eqref{eq:node_intervention_12} defines the process to swap the feature from $G_B$ to $G_A$. The final decision depends on the process with less cost (\Eqref{eq:node_intervention_d}). We refer to this cost as the \textit{intervention cost} between $G_A$ and $G_B$. Finally, we add $(\vx^*, \vt^*)$ to the dataset $\mathcal D$ and collect enough user clicks about it, which connects $G_A$ and $G_B$ together.

To extend the above process of connecting two components to the entire IG, consider a graph $G$ comprising $K$ connected components: $G = G_1 \cup \cdots \cup G_K$. We first construct another complete graph with $K$ nodes, each representing a component in $G$. Intervention costs between components $G_i$ and $G_j$ serve as edge weights in the constructed complete graph. A Minimum Spanning Tree (MST) algorithm is applied to this complete graph to find edges with minimal total weights for connection. Subsequent interventions (\Eqref{eq:node_intervention_cost} - \Eqref{eq:node_intervention_d}) are performed on these edges to connect the entire IG with the minimal total intervention cost. The comprehensive algorithm is detailed in Appendix \ref{sec:app_algo_ni}. 

Compared to traditional intervention strategies that often necessitate random swapping across all queries \citep{joachims2017unbiased,radlinski2006minimally,carterette2018offline,yue2010beyond,wang2018position}, node intervention requires only $K - 1$ swaps for an IG with $K$ connected components. Notably, $K$ is generally much smaller than the total query count, particularly when positions are the sole bias factors, which is the usual focus of traditional intervention strategies. Consequently, node intervention markedly reduces online interventions and improves the user experience.

\subsection{Node merging}
\label{sec:node_merging}

Despite node intervention being effective in achieving identifiability, it still requires additional online experiments, which can be time-consuming and may pose a risk of impeding user experience by displaying irrelevant documents at the top of the ranking list. What's worse, some types of bias factors may not be appropriate for swapping (\eg, contextual information or other documents' clicks). Therefore, we propose another simple and general methodology for addressing the unidentifiability issue, which involves merging nodes from different connected components and forcing them to have the same observation prediction. We refer to this strategy as \textbf{node merging}, which is illustrated in Figure \ref{fig:method_nm}.

Similar to node intervention, there are numerous options for selecting node pairs to merge. Note that merging two dissimilar nodes with distinct observation probabilities will inevitably introduce approximation errors, as stated in the following proposition (We defer the proof to Appendix \ref{sec:app_proof_4}):

\begin{proposition}[Error bound of merging two components]
\label{thm:merging_bound}
Suppose an IG $G=(V, E)$ consists of two connected components $G_1 = (V_1, E_1)$ and $G_2 = (V_2, E_2)$. If we merge two nodes $v_1\in G_1$ and $v_2\in G_2$ by forcing $o'(\vt') = o'(\vt'')$ where $v_1$ and $v_2$ represent bias factors $\vt'$ and $\vt''$, then:
\begin{align*}
    &r(\vx)\cdot o(\vt) = r'(\vx)\cdot o'(\vt)\\
    & \Longrightarrow  \left|\frac{r'(\vx_1)}{r(\vx_1)} - \frac{r'(\vx_2)}{r(\vx_2)}\right| \leq \left|\frac{o(\vt') - o(\vt'')}{o'(\vt')}\right|, \forall \vx_1, \vx_2\in\mathcal X,
\end{align*}
where we suppose $r, r', o$ and $o'$ are not zero.
\end{proposition}
\begin{remark}
    When $o(\vt') = o(\vt'')$, the relevance model is identifiable. If the gap between $o(\vt')$ and $o(\vt'')$ is large, $\nicefrac{r'(\vx)}{r(\vx)}$ will change greatly, hurting the performance.
\end{remark}

Therefore, we propose to merge similar nodes exhibiting minimal differences in their observation probabilities. Suppose each bias factor $\vt$ can be represented using a feature vector $\mX_{\vt}$ (namely bias features). We further assume that vectors with greater similarity correspond to closer observation probabilities. As a simple example, we can consider the number of positions as a 1-dimensional bias feature, since it is reasonable that documents in proximate positions will exhibit similar observation probabilities.

\begin{table*}[tbp]
    \centering
    \small
    \caption{Performance of different methods on the $K=2$ simulation dataset under PBM bias. We ran each experiment 10 times and reported the average results as well as the standard deviations.}
    \fontsize{8pt}{8pt}\selectfont
\begin{tabular}{lcccccc}
\toprule
    \textbf{Method} & \textbf{MCC}$\:\uparrow$ & \textbf{nDCG@1}$\:\uparrow$ & \textbf{nDCG@3}$\:\uparrow$ & \textbf{nDCG@5}$\:\uparrow$ & \textbf{nDCG@10}$\:\uparrow$ & \textbf{Click MSE} \\
\midrule
No debias & $0.521 _{\pm .000}$ & $0.711 _{\pm .000}$ & $0.625 _{\pm .000}$ & $0.665 _{\pm .000}$ & $0.820 _{\pm .000}$ & $2\times 10^{-5}$ \\
\arrayrulecolor{gray!80}
\midrule
\arrayrulecolor{black}
DLA & $0.707 _{\pm .105}$ & $0.836 _{\pm .061}$ & $0.742 _{\pm .091}$  & $0.789 _{\pm .070}$ & $0.886 _{\pm .040}$ & $<10^{-8}$ \\
\quad + Node intervention & $\mathbf{1.000 _{\pm .000}}$ & $\mathbf{1.000 _{\pm .000}}$ & $\mathbf{1.000 _{\pm .000}}$ & $\mathbf{1.000 _{\pm .000}}$ & $\mathbf{1.000 _{\pm .000}}$ & $<10^{-8}$ \\
\quad + Node merging & $0.975 _{\pm .000}$ & $\mathbf{1.000 _{\pm .000}}$ & $\mathbf{1.000 _{\pm .000}}$ & $\mathbf{1.000 _{\pm .000}}$ & $\mathbf{1.000 _{\pm .000}}$ & $<10^{-8}$ \\
\arrayrulecolor{gray!80}
\midrule
\arrayrulecolor{black}
Regression-EM & $0.580_{\pm .117}$  &$0.786_{\pm .035}$ & $0.677_{\pm .063}$ & $0.752_{\pm .044}$&  $0.857_{\pm .027}$ & $<10^{-8}$ \\
\quad + Node intervention & $\mathbf{0.980_{\pm .023}}$  & ${0.999_{\pm .001}}$  & $0.995_{\pm .010}$  & $0.989_{\pm .023}$   &$0.997_{\pm .006}$ & $<10^{-7}$ \\
\quad + Node merging & $0.975_{\pm .000}$&  $\mathbf{1.000_{\pm .000}}$ & $\mathbf{1.000_{\pm .000}}$&$\mathbf{1.000_{\pm .000}}$&  $\mathbf{1.000_{\pm .000}}$ & $<10^{-8}$ \\
\arrayrulecolor{gray!80}
\midrule
\arrayrulecolor{black}
Two-Tower & $0.830 _{\pm .050}$ & $0.883 _{\pm .034}$ & $0.832 _{\pm .054}$  & $0.857 _{\pm .045}$ & $0.925 _{\pm .022}$ & $<10^{-8}$ \\
\quad + Node intervention & $\mathbf{1.000 _{\pm .000}}$ & $\mathbf{1.000 _{\pm .000}}$ & $\mathbf{1.000 _{\pm .000}}$ & $\mathbf{1.000 _{\pm .000}}$ & $\mathbf{1.000 _{\pm .000}}$ & $<10^{-8}$ \\
\quad + Node merging & $0.975 _{\pm .000}$ & $\mathbf{1.000 _{\pm .000}}$ & $\mathbf{1.000 _{\pm .000}}$ & $\mathbf{1.000 _{\pm .000}}$ & $\mathbf{1.000 _{\pm .000}}$ & $<10^{-8}$ \\
\bottomrule
\end{tabular}%

  \label{tab:experiment_simulation_k2}%
\end{table*}%

Based on it, we use the following process to connect two components $G_A=(V_A, E_A)$ and $G_B=(V_B, E_B)$:
\begin{align}
    \label{eq:node_merge_cost}
    \mathcal C(\vt_1, \vt_2) &= ||\mX_{\vt_1} - \mX_{\vt_2}||,\\
    \label{eq:node_merge_t12}
    \vt_A^*, \vt_B^* &= {\arg\min}_{t_A\in V_A, t_B \in V_B} 
    \mathcal C(\vt_A, \vt_B).
\end{align}
Here \Eqref{eq:node_merge_cost} defines the merging cost to merge $\vt_1$ and $\vt_2$. \Eqref{eq:node_merge_t12} finds two bias factors $\vt_A^*$ and $\vt_B^*$ from two components that have the minimal merging cost. We refer to this cost as the merging cost between $G_A$ and $G_B$. Analogous to node intervention, the MST algorithm is employed to connect the IG, with the sole distinction being the definition of edge weights in the complete graph as merging costs. The comprehensive algorithm is detailed in Appendix \ref{sec:app_algo_nm}. 

Furthermore, we present two key properties of node merging in Appendix \ref{sec:error_bound_merging}: (1) \textit{Consistency}: the merging constraints imposed on $o'$ are compatible with the preconditions for identifiability (\ie, $r(\vx)\cdot o(\vt)=r'(\vx)\cdot o'(\vt)$), ensuring that click probabilities can still be accurately fitted after applying node merging; and (2) \textit{Error bound}: extended from Proposition \ref{thm:merging_bound}, the error for node merging is bounded by the diameter of the constructed MST.

Compared to node intervention, node merging performs on the offline dataset, making it a simple and time-efficient approach. However, merging bias factors brings additional approximation error which has the potential to adversely impact the ranking performance. 
\section{Experiments}

In this section, we describe our experimental setup and show the empirical results, in both the fully synthetic setting and large-scale study. \footnote{Code is available at \url{https://github.com/Keytoyze/ULTR-identifiability}}

\subsection{Fully synthetic study}

\paragraph{Dataset} To verify the correctness of Theorem \ref{thm:graph}, and the effectiveness of proposed methods, we first conducted experiments on a fully synthetic dataset, which allowed for precise control of the connectivity of IGs. We generated four datasets with different numbers $K$ of connected components within each IG ($K=1,2,3,4$), as illustrated in Appendix \ref{sec:app_datasets}. The bias factors only consist of positions (\ie, position-based model or PBM). We defer the click simulation setup to Appendix \ref{sec:app_click_simulation}.

\paragraph{Baselines} For comparative analysis, we evaluated our methods on several baselines: \textit{No debias}, which trains the ranking model using click data without an observation model, and three widely-used ULTR optimization algorithms based on examination hypothesis, \textit{DLA} \citep{ai2018unbiased}, \textit{Regression-EM} \citep{wang2018position} and \textit{Two-Tower} \citep{guo2019pal}. Notably, many recent models \citep{vardasbi2020cascade,chen2021adapting,sarvi2023impact,cacm,chen2022lbd,chen2023multi} are variants of these three ULTR algorithms, primarily varying in their bias factor handling. We will discuss the influence of bias factors on identifiability in the next section (\sref{sec:experiment_large}). Training details for the baselines can be found in Appendix \ref{sec:app_training}.  Our \textit{node intervention} and \textit{node merging}, being model-agnostic and training-independent, are applied to datasets before training.

\paragraph{Evaluation metrics} To evaluate the performance
of the methods, we computed the mean correlation coefficient (\textbf{MCC}) between the true relevance probability $r(\cdot)$ and the predicted relevance probability $r'(\cdot)$, defined as
\begin{align*}
    \frac{
        \sum_{i=1}^{|\mathcal D|} \left(
            r(\vx_i)-\overline{r(\vx)}
        \right)\left(
            r'(\vx_i)-\overline{r'(\vx)}
        \right)
    }{
        \sqrt{
            \sum_{i=1}^{|\mathcal D|} \left(
                r(\vx_i)-\overline{r(\vx)}
            \right)^2
        }
        \sqrt{
            \sum_{i=1}^{|\mathcal D|} \left(
                r'(\vx_i)-\overline{r'(\vx)}
            \right)^2
        }
    }.
\end{align*}

A high MCC means that we successfully identified the true model and recovered the true relevance up to a scaling transformation. We also computed \textbf{nDCG}, which are standard ranking metrics prevalently used in LTR, and \textbf{Click MSE}, which is the mean squared error between the true and predicted click probability for evaluating the fitting goodness.

\paragraph{Analysis: How does the connectivity of IGs impact the ranking performance?} Figure \ref{fig:dataset_components} shows the effects of varying numbers of connected components $K$ within IGs on ranking performance (using DLA), with different numbers of clicks. Here, $K=1$ indicates a connected IG. We can observe that the ranking model is capable of achieving perfect ranking accuracy only if the IG is connected. Otherwise, the performance exhibits significant instability and poor quality, regardless of the number of clicks collected. 
Besides, larger $K$s (\eg, when $K>1$) do not give rise to significant differences. 
These observations serve to validate the correctness of Theorem \ref{thm:graph}.

\begin{figure}[tbp]
    \centering%
    \subfigure[DLA]{
        \includegraphics[width=0.23\textwidth,trim={10 10 10 10},clip]{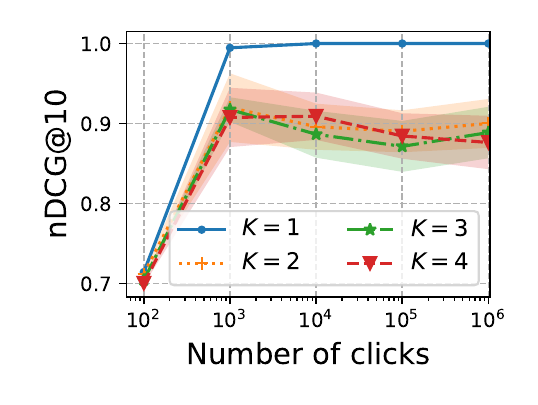}
        \label{fig:dataset_components}
    }%
    \subfigure[Different methods]{
        \includegraphics[width=0.23\textwidth,trim={10 3 7 20},clip]{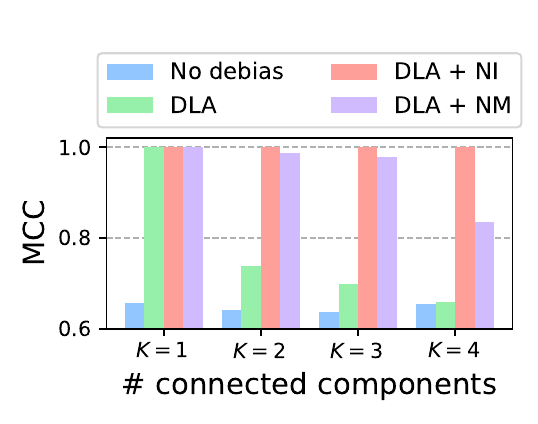}
        \label{fig:dataset_components_method}
    }
    \caption{(a) Performance of DLA on different numbers of connected components $K$ across different click counts. Shadowed areas depict variance. (b) Performance of different methods across different $K$. NI = \underline{N}ode \underline{I}ntervention. NM = \underline{N}ode \underline{M}erging. }
\end{figure}
\begin{figure*}[tbp]
    \centering
        \subfigure[Node intervention]{
        \includegraphics[width=0.24\textwidth,trim={10 10 10 10},clip]{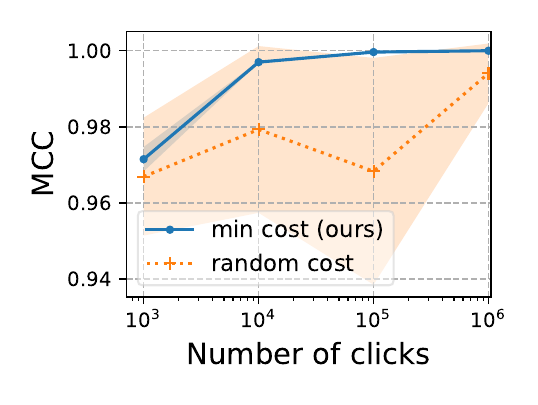}
        \label{fig:simulation_ni}
    }%
    \subfigure[Node merging]{
    \includegraphics[width=0.24\textwidth,trim={10 10 10 10},clip]{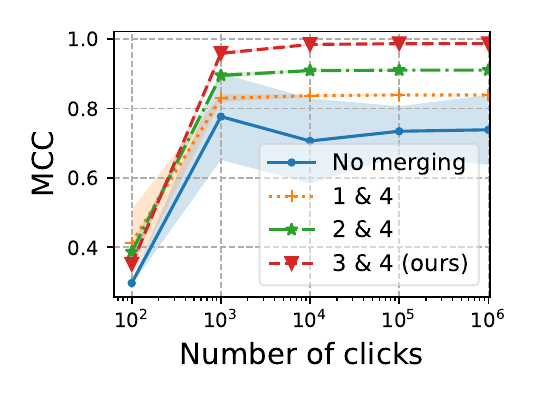}
        \label{fig:simulation_nm}
    }%
    \subfigure[Yahoo!]{
    \includegraphics[width=0.24\textwidth,trim={10 10 10 10},clip]{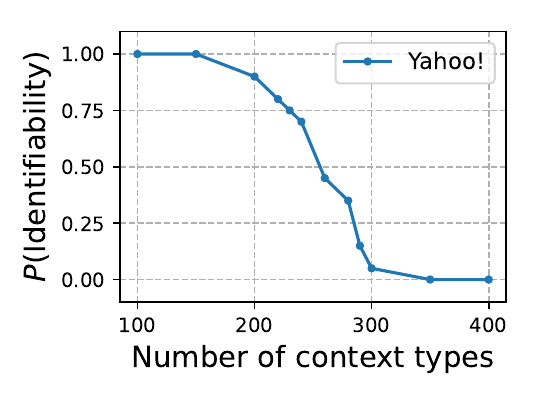}
        \label{fig:sample_context_yahoo}
    }%
    \subfigure[Istella-S]{
    \includegraphics[width=0.24\textwidth,trim={10 10 10 10},clip]{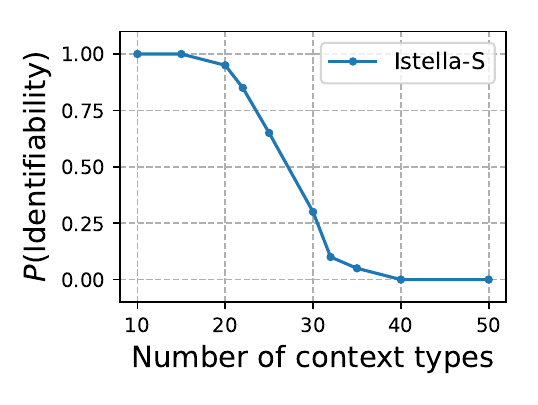}
        \label{fig:sample_context_istella}
    }
    \caption{(a) Impact of cost selection strategies in node intervention. (b) Impact of merging strategies in node merging. (c)(d) Impact of the number of context types on the identifiability probabilities on the two datasets.}
\end{figure*}

\paragraph{Analysis: Can our proposed two methods handle unidentifiable datasets?} We tested the different methods in the $K=2$ scenario and summarized the results in Table \ref{tab:experiment_simulation_k2}. Our methods consistently achieved nearly perfect ranking accuracy, markedly surpassing the baselines. Notably, various ULTR algorithms tend to yield disparate suboptimal performances in unidentifiable cases, despite perfectly fitting click probability as indicated by Click MSE. However, with identifiability established through our approaches, they all converge to a common set of effective parameters that accurately recover relevance. Furthermore, node intervention recovers a more accurate relevance than node merging, evidenced by the MCC. This shows the trade-offs between the two proposed methods: compared to node merging, node intervention does not introduce approximation errors, but requires an additional step to augment the dataset.

\paragraph{Analysis: How does the methods' performance degrade with decreased IG connectivity?}  We evaluated DLA and our approaches under varying the value of $K$. Figure \ref{fig:dataset_components_method} illustrates that both DLA and node merging show performance declines as $K$ increases, while node merging exhibits a slower rate. Remarkably, node intervention sustains a perfect relevance recovery ability. We also noted the node merging performance for $K=4$ is suboptimal compared to $K=3$ and $K=2$, attributed to merging two relatively divergent nodes in $K=4$. A comparative analysis of observation estimation in Appendix \ref{sec:observation} provides an in-depth discussion of this suboptimal performance.

\paragraph{Ablation study: Impact of different selection strategies for node intervention.} In our ablation study of the node intervention method within the $K=2$ scenario (using DLA), we explored a variation termed \textit{random cost}. Here, the cost function (\Eqref{eq:node_intervention_cost}) was modeled to follow a uniform distribution over $[0, 1]$ and unrelated to $\vx$ and $\vt$, which leads to random selection of intervention pairs (\Eqref{eq:node_intervention_21}-\Eqref{eq:node_intervention_d}). The original method is denoted as \textit{min cost}.  As observed in Figure \ref{fig:simulation_ni}, the \textit{random cost} approach exhibits notably higher variance than \textit{min cost} and requires sufficient clicks to obtain a stable performance, which confirms the validity of Proposition \ref{thm:variance}. Additionally, despite its variance, \textit{random cost} attains a commendable performance level relative to the unidentifiable baseline, reinforcing the significance of maintaining a connected IG.

\paragraph{Ablation study: Impact of different merging strategies for node merging.} Similarly, we conducted an ablation study for the node merging method, based on DLA within the $K=2$ scenario. We use three different merging strategies, where $a$ \& $b$ represents merging nodes corresponding to the position $a$ and $b$: 1 \& 4, 2 \& 4, and 3 \& 4. Note that all of the strategies ensure a connected IG (See Figure \ref{fig:experiment_simulation_ig} in Appendix \ref{sec:app_datasets} for details), and 3 \& 4 is the proposed node merging strategy. Figure \ref{fig:simulation_nm} demonstrates that performance improves as the merging nodes are closer, thereby validating Proposition \ref{thm:merging_bound}.

\subsection{Large-scale study}

\label{sec:experiment_large}



\begin{table*}[t]
  \centering
  \small
  \caption{Performance (with the standard deviations) comparison on two datasets under CPBM bias.}
  \fontsize{8pt}{8pt}\selectfont

\begin{tabular}{clccccccc}
\toprule
\multirow{2}[3]{*}{\textbf{Dataset}} & \multicolumn{1}{c}{\multirow{2}[3]{*}{\textbf{Method}}} & \multicolumn{4}{c}{\textbf{Training}} &      & \multicolumn{2}{c}{\textbf{Test}} \\
\arrayrulecolor{gray!80}\cmidrule{3-6}\cmidrule{8-9}\arrayrulecolor{black}     &      & \textbf{MCC}$\:\uparrow$  & \textbf{nDCG@5}$\:\uparrow$ & \textbf{nDCG@10}$\:\uparrow$ & \textbf{Click MSE} &      & \textbf{nDCG@5}$\:\uparrow$ & \textbf{nDCG@10}$\:\uparrow$ \\
\midrule
\multirow{3}[2]{*}{Yahoo} & No debias & $0.765_{\pm.000}$ & $0.841_{\pm.000}$ & $0.915_{\pm.000}$ & $4\times 10^{-4}$ &      & $0.693_{\pm.002}$ & $0.741_{\pm.001}$ \\
     & DLA  & $0.750_{\pm.000}$ & $0.844_{\pm.000}$ & $0.914_{\pm.000}$ & $2\times 10^{-5}$ &      & $0.693_{\pm.001}$ & $0.741_{\pm.001}$ \\
     &  DLA + Node merging & \boldmath{}\textbf{$0.771_{\pm.000}$}\unboldmath{} & \boldmath{}\textbf{$0.853_{\pm.000}$}\unboldmath{} & \boldmath{}\textbf{$0.920_{\pm.000}$}\unboldmath{} & $4\times 10^{-5}$ &      & \boldmath{}\textbf{$0.697_{\pm.001}$}\unboldmath{} & \boldmath{}\textbf{$0.745_{\pm.001}$}\unboldmath{} \\
 \arrayrulecolor{gray!80}
\midrule
\arrayrulecolor{black}
\multirow{3}[2]{*}{Istella-S} & No debias & $0.764_{\pm.000}$ & $0.885_{\pm.000}$ & $0.941_{\pm.000}$ & $4\times 10^{-5}$ &      & $0.634_{\pm.001}$ & $0.682_{\pm.001}$ \\
     & DLA  & $0.764_{\pm.000}$ & $0.886_{\pm.000}$ & $0.941_{\pm.000}$ & $1\times 10^{-6}$ &      & $0.633_{\pm.001}$ & $0.682_{\pm.001}$ \\
     &  DLA + Node merging & \boldmath{}\textbf{$0.772_{\pm.000}$}\unboldmath{} & \boldmath{}\textbf{$0.892_{\pm.000}$}\unboldmath{} & \boldmath{}\textbf{$0.944_{\pm.000}$}\unboldmath{} & $2\times 10^{-5}$ &      & \boldmath{}\textbf{$0.636_{\pm.001}$}\unboldmath{} & \boldmath{}\textbf{$0.684_{\pm.001}$}\unboldmath{} \\
\bottomrule
\end{tabular}%

  \label{tab:experiment_real}%
\end{table*}%

\paragraph{Dataset} We also performed another empirical study on the large-scale semi-synthetic setup that is prevalent in unbiased learning to rank \citep{joachims2017unbiased,ai2021unbiased,chen2022lbd} on two widely used benchmark datasets: Yahoo! LETOR \citep{chapelle2011yahoo}  and Istella-S \citep{lucchese2016post}.  We provide further details for them in Appendix \ref{sec:app_datasets}. In both datasets, only the top 10 documents were considered to be
displayed. In addition to positions, we also incorporated context types as another bias factor that is prevalent in recent research (\ie, contextual position-based model or CPBM)  \citep{fang2019intervention}. A random context type was assigned to each query-document pair. Furthermore, we conducted identifiability testing on the TianGong-ST \citep{chen2019tian}, a large-scale real-world dataset with an abundance of genuine bias factors.

\paragraph{Analysis: Does the unidentifiability issue exist in the real world?} We first applied the identifiability check algorithm on TianGong-ST and found that when accounting for \textbf{positions} and all provided \textbf{vertical types} as bias factors (a total of 20,704), the IG of this dataset is \textit{disconnected}: there are 2,900 connected components within it. This observation suggests that the unidentifiability phenomenon could occur in the real world, even on a large-scale dataset. We further excluded certain bias factors from that dataset and assessed its identifiability, which is elaborated in Appendix \ref{sec:app_tiangong}.

\paragraph{Analysis: How do the number of bias factors $|\mathcal T|$, dataset scale $|\mathcal D|$ and feature count $|\mathcal X|$ affect the identifiability?} We tuned the number of context types within Yahoo! and Istella-S and computed the frequency with which the IG was connected to determine the identifiability probability. From Figure \ref{fig:sample_context_yahoo} and \ref{fig:sample_context_istella}, it can be observed that if positions are the only bias factors, both datasets are identifiable for the ranking model. However, upon the consideration of context types, the identifiability probability drastically decreases as the number of context types increases. Merely 20 (on Istella-S) or 200 (on Yahoo!) are sufficient to render the IG disconnected. When the number is further increased, it becomes exceedingly challenging to obtain an identifiable ranking model, which aligns with the conclusion of Example \ref{thm:identifiability_prob}. We also experimented with investigating the impact of dataset scale and feature count on identifiability, elaborated in Appendix \ref{sec:app_sampling_ratio}.

\paragraph{Analysis: How does node merging perform under contextual bias?} We simulated 5,000 context types on Yahoo! and Istella-S to evaluate the efficacy of node merging. To eliminate the influence of initialization, all predicted observation probabilities are initially set to 1.0. In such cases, node intervention is not practical since context types cannot be swapped. Results in both the training and test sets in Table \ref{tab:experiment_real} show that node merging correctly manages unidentifiable cases and restores relevance, significantly surpassing the baselines in MCC and nDCG. 


\paragraph{Analysis: How does model initialization impact performance?} In unidentifiable scenarios, models can converge to various parameters that predict click probability correctly, but only a limited range truly represents correct relevance. Therefore, model initialization notably affects performance in unidentifiable scenarios. We detail it in Appendix \ref{app:initialization}.


\section{Related work}

\label{sec:related-work}

\paragraph{Unbiased learning to rank (ULTR)}  Unbiased learning to rank (ULTR) tries to directly learn unbiased ranking models from biased clicks. The core of ULTR lies in the estimation of observation probabilities, which is typically achieved through intervention \citep{wang2016learning,joachims2017unbiased}. These methods are related to the node intervention we proposed in \sref{sec:node_intervention}, but they are prone to useless swapping operations and negatively impact the user's experience. To avoid intervention, \citet{agarwal2019estimating} and \citet{fang2019intervention} proposed intervention harvest methods that exploit click logs with multiple ranking models. These methods are related to Theorem \ref{thm:graph}, but they did not delve into the identifiability conditions. Our proposed theory bridges the gap between these two groups of work by determining when intervention is necessary or not. Recently, some researchers proposed to jointly estimate relevance and bias, including IPS-based methods \citep{wang2018position,ai2018unbiased,hu2019unbiased,jin2020deep} and two-tower based models \citep{zhao2019recommending,guo2019pal,haldar2020improving,yan2022revisiting}. These models are based on the examination hypothesis and are optimized to maximize user click likelihood, therefore our proposed framework can be applied to these models as well.

On the other side, researchers developed models to extend the scope of bias factors that affect observation probabilities, which contain position \citep{wang2018position,ai2018unbiased,hu2019unbiased,ai2021unbiased}, contextual information \citep{fang2019intervention,tian2020counterfactual}, clicks in the same query list \citep{vardasbi2020cascade,chen2021adapting}, presentation style \citep{zheng2019constructing,liu2015influence,chen2023multi}, search intent \citep{sun2020eliminating}, result domain \citep{ieong2012domain}, ranking features \citep{chen2022lbd} and outliers \citep{sarvi2023impact}. While incorporating additional bias factors is beneficial for improving the estimation of accurate observation probabilities \citep{chen2023multi}, as we mention in Example \ref{thm:identifiability_prob} and \sref{sec:experiment_large}, an excessive number of bias factors may pose a risk of unidentifiability.

\paragraph{Relevance recovery in ULTR} The identifiability condition (Theorem \ref{thm:graph}) establishes connections and generalizations to recent ULTR research. \citet{agarwal2019estimating} constructed intervention sets to uncover documents that are put at two different positions to estimate observation probabilities, which, however, did not further explore the recoverable conditions. \citet{oosterhuis2022reaching} also showed that a perfect click model can provide incorrect relevance estimates and the estimation consistency depends on the data collection procedure. We take a further step by delving into the root cause and digesting the concrete condition based on the data. \citet{zhang2023towards} found that some features (\eg, with high relevance) are more likely to be assigned with specific bias factors (\eg, top positions). This phenomenon results in a decline in performance, called confounding bias. This bias is related to the identifiability issue since when a severe confounding bias is present, the IG is more likely to be disconnected due to insufficient data coverage.

\paragraph{Identifiability} Identifiability is a fundamental concept in various machine learning fields, such as independent component analysis \citep{hyvarinen2017nonlinear,hyvarinen2019nonlinear}, latent variable models \citep{allman2009identifiability,guillaume2019introductory,khemakhem2020variational}, missing not at random data \citep{ma2021identifiable,miao2016identifiability} and causal discovery \citep{addanki2021intervention,peters2011identifiability,spirtes2016causal}. It defines a model's capacity to recover some unique latent structure, variables, or causal relationships from the data. In this work, we embrace the commonly used notion of identifiability and apply its definition to the ULTR domain.
\section{Conclusions}

In this paper, we take the first step to exploring if and when relevance can be recovered from click data from a foundational perspective. We first define the identifiability of a ranking model, which refers to the ability to recover relevance probabilities up to a scaling transformation. Our research reveals that (1) the ranking model is not always identifiable, which depends on the underlying structure of the dataset (\ie, an identifiability graph should be connected); (2) identifiability depends on the data size, the number of bias factors and features, and unidentifiability issue is possible on large-scale real-world datasets; (3) two methods, node intervention and node merging, can be utilized to address the unidentifiability issues. Our proposed framework is theoretically and empirically verified.

\paragraph{Limitations and future work} (1) While examination hypothesis (\Eqref{eq:examination_hypothesis}) is the most widely used hypothesis, other models like trust bias \citep{agarwal2019addressing} and vector-based examination hypothesis \citep{chen2022scalar,yan2022revisiting} also merit attention. Adapting our graph-based approach to them is a promising area for future research. (2) While the assumption of perfect data fitting is typical for identifiability theory in machine learning \cite{hyvarinen2019nonlinear,khemakhem2020variational}, investigating the propagation of approximation errors within the identifiability graph is an interesting future direction in ULTR.

\section*{Impact statement}
This paper presents work whose goal is to advance the field of Machine Learning. There are many potential societal consequences of our work, none of which we feel must be specifically highlighted here.

\section*{Acknowledgments}
Research work mentioned in this paper is supported by State Street Zhejiang University Technology Center. We would also like to thank Lefei Shen for the assistance in the experiments and thank Yusu Hong, Yu Mou, Shanda Li, and Wencheng Cai for the valuable discussion of the theory.



\bibliography{ref}
\bibliographystyle{icml2024}

\newpage
\appendix
\onecolumn
\numberwithin{equation}{section}

\section*{Appendix}

\section{Theoretical results}

\label{sec:app_proof}

\subsection{Proof for Theorem \ref{thm:graph}}

\label{sec:app_proof_1}

\textbf{Step 1.} We first prove the "if" part: assume that 
\begin{align}
    \label{eq:thm1_click_equal}
    r(\vx)\cdot o(\vt)=r'(\vx)\cdot o'(\vt)\quad  \forall (\vx,\vt)\in \mathcal D,
\end{align}
and $G$ is connected, our goal is to prove that $\nicefrac{r(\vx)}{r'(\vx)}=\text{constant}$. Note that we only consider the nontrivial case $r(\vx)\neq 0$ and $r'(\vx) \neq 0$. Otherwise, $C$ can be any positive number. 

For any two bias factors $\vs\in\mathcal T$ and $\vt \in \mathcal T$, since $G$ is connected, there exists a path $v_{\va_1} \rightarrow v_{\va_2} \rightarrow \cdots \rightarrow v_{\va_n}$ in $G$ where $v_{\va_1},\cdots,v_{\va_n}$ are the nodes in the identifiability graph representing different bias factors, and $\va_1=\vs, \va_n=\vt$. Consider a middle edge $v_{\va_m}\rightarrow v_{\va_{m+1}} (1\leq m\leq n-1)$, according to the definition of the edge,
\begin{align}
\label{eq:thm1_edge}
\exists \vx\in \mathcal X, \text{ s.t. } (\vx, \va_m)\in \mathcal D \land (\vx,\va_{m+1})\in \mathcal D.
\end{align}

According to \Eqref{eq:thm1_click_equal} and \Eqref{eq:thm1_edge}, we have $r(\vx)\cdot o(\va_m)=r'(\vx)\cdot o'(\va_m)$ and $r(\vx)\cdot o(\va_{m+1})=r'(\vx)\cdot o'(\va_{m+1})$, and therefore 
\begin{align}
\label{eq:thm1_edge_equal}
\frac{o'(\va_m)}{o(\va_m)}=\frac{r(\vx)}{r'(\vx)}=\frac{o'(\va_{m+1})}{o(\va_{m+1})}.
\end{align}

Let $f(\vx)=\nicefrac{r(\vx)}{r'(\vx)}$ and $g(\vt)=\nicefrac{o'(\vt)}{o(\vt)}$. Applying \Eqref{eq:thm1_edge_equal} to the path $v_{\va_1} \rightarrow v_{\va_2}\rightarrow \cdots \rightarrow v_{a_n}$, we obtain $g(\vs)=g(\vt)$. Given that $\vs$ and $\vt$ are selected arbitrarily, we have $g(\vt)=\text{constant}$ for all bias factors $\vt$. Since $f(\vx)=g(\vt)$ holds for all $(\vx, \vt)\in\mathcal D$ according to \Eqref{eq:thm1_click_equal}, $f(\vx)$ is also constant.

\textbf{Step 2.} We then prove the "only if" part: assume that the relevance is identifiable, and prove that $G$ is connected. We prove this by contradiction: Given a disconnected IG $G$, our goal is to prove that the ranking model is unidentifiable, by showing that we can construct two click models such that the click probabilities are equal, yet the inside relevance models differ.

Since $G=(V, E)$ is disconnected, we suppose $G$ can be divided into two disjoint graphs $G_1=(V_1, E_1)$ and $G_2=(V_2, E_2)$. Note that we do not require $G_1$ or $G_2$ to be connected, therefore this division is always feasible even when $G$ has more than two components. Based on $G_1$ and $G_2$, we can divide the dataset $\mathcal{D}$ into two disjoint sets $D_1$ and $D_2$: $\mathcal{D}_1=\{(\vx, \vt) \mid v_{\vt}\in V_1\}$ and $\mathcal{D}_2=\{(\vx, \vt) \mid v_{\vt}\in V_2\}$. Let $\mathcal X_1=\{\vx\mid(\vx, \vt)\in\mathcal{D}_1\}$ denote features in $\mathcal{D}_1$, and $\mathcal X_2=\{\vx\mid(\vx, \vt)\in\mathcal{D}_2\}$ denote features in $\mathcal{D}_2$. Note that $\mathcal X_1$ and $\mathcal X_2$ are disjoint, \ie, $\mathcal X_1 \cap \mathcal X_2=\varnothing$, otherwise according to the definition of the edge set, there exists an edge between $V_1$ and $V_2$ which connects $G_1$ and $G_2$.

Next, given any relevance function $r$ and observation function $o$, we define $r'$ and $o'$ as follows.
\begin{align*}
    r'(\vx) &= \begin{cases}
        \alpha r(\vx) & \mbox{if $\vx\in\mathcal X_1$,}\\
        \beta r(\vx) & \mbox{if $\vx\in\mathcal X_2$,}
    \end{cases} \\
    o'(\vt) &= \begin{cases}
        \nicefrac{o(\vt)}{\alpha} & \mbox{if $v_{\vt}\in V_1$,}\\
        \nicefrac{o(\vt)}{\beta} & \mbox{if $v_{\vt}\in V_2$,}
    \end{cases}
\end{align*}
where $\alpha\neq \beta$ are two positive numbers. Note that if $(\vx, \vt)\in\mathcal{D}_1$, then $\vx\in\mathcal{X}_1$ and $v_{\vt}\in V_1$, therefore $r'(\vx)o'(\vt)=\alpha r(\vx)\cdot \nicefrac{o(\vt)}{\alpha} = r(\vx)o(\vt)$. If $(\vx, \vt)\in\mathcal{D}_2$, then $\vx\in\mathcal{X}_2$ and $v_{\vt}\in V_2$, therefore $r'(\vx)o'(\vt)=\beta r(\vx)\cdot \nicefrac{o(\vt)}{\beta } = r(\vx)o(\vt)$. Based on it, \Eqref{eq:thm1_click_equal} holds for all $(\vx, \vt)\in\mathcal D$. However, it is obvious that $C$ isn't constant in $r(\vx)=C r'(\vx)$, since $C=\alpha$ when $\vx\in\mathcal X_1$ and $C=\beta$ otherwise. It indicates that the relevance model isn't identifiable.



\subsection{Proof for Example \ref{thm:identifiability_prob}}

\label{sec:app_proof_2}

We begin by estimating the disconnected probability between two nodes in the identifiability graph, as the following lemma.

\begin{lemma}
\label{lem:link_prob}
In an identifiability graph, the probability of two nodes $v_{\vs}$ and $v_{\vt}$ are disconnected can be estimated as:
\begin{align*}
    P(\text{disconnected}\mid \mathcal D, \vs, \vt) \sim \exp\left(
        -\frac{|\mathcal D|}{|\mathcal {T}|}
    \right)\left[
        2 - \exp\left(
            -\frac{|\mathcal D|}{|\mathcal X||\mathcal {T}|}
        \right)
    \right]^{|\mathcal X|},
\end{align*}
when $|\mathcal X| |\mathcal T| \rightarrow \infty$.
\end{lemma}

\begin{proof}

Let $P(\vx, \vs)$ and $P(\vx, \vt)$ denote the probabilities of selecting $(\vx, \vs)$ and $(\vx, \vt)$ respectively. We have:

\begin{align*}
P(\text{disconnected}\mid \mathcal D, \vs, \vt) &= P\left(
\bigcap_{\vx\in \mathcal X} (\vx, \vs)\notin \mathcal D \lor (\vx, \vt)\notin \mathcal D
\right)\\
&= \prod_{\vx\in\mathcal X} P\left(
    (\vx, \vs)\notin \mathcal D \lor (\vx, \vt)\notin \mathcal D
\right) \\
&= \prod_{\vx\in\mathcal X} 1 - P\left(
    (\vx, \vs)\in \mathcal D
\right)\cdot P\left(
    (\vx, \vt)\in \mathcal D
\right) \\
&= \prod_{\vx\in\mathcal X} 1 - \left(
    1 - P\left( (\vx, \vs) \notin \mathcal{D} \right)
\right)\cdot \left(
    1 - P\left( (\vx, \vt) \notin \mathcal{D} \right)
\right).
\end{align*}

Note that $P\left( (\vx, \vt) \notin \mathcal{D} \right)$ is the probability that $(\vx, \vt)$ is not sampled for $|\mathcal D|$ times, we have $P\left( (\vx, \vs) \notin \mathcal{D} \right) = \left[1 - P(\vx, \vs)\right]^{|\mathcal D|}$ and $P\left( (\vx, \vt) \notin \mathcal{D} \right) = \left[1 - P(\vx, \vt)\right]^{|\mathcal D|}$, therefore,
\begin{align*}
    P(\text{disconnected}\mid \mathcal D, \vs, \vt) = \prod_{\vx\in\mathcal X} 1 - \left(
        1 - \left[
            1 - P(\vx, \vs)
        \right]^{|\mathcal D|}
    \right) \left(
        1 - \left[
            1 - P(\vx, \vt)
        \right]^{|\mathcal D|}
    \right).
\end{align*}


Using the condition that features and bias factors are sampled independently and uniformly, we have $P(\vx, \vs) = P(\vx, \vt) = \nicefrac{1}{|\mathcal X||\mathcal T|}$. Therefore,
\begin{align*}
P(\text{disconnected}\mid \mathcal D, \vs, \vt) 
&= 
\left\{ 1 - \left[
    1 - \left(
        1 - \frac{1}{|\mathcal X||\mathcal T|}
    \right)^{|\mathcal D|}
\right]^2\right\}^{|\mathcal X|}\\
&= \left\{ 1 - \left[
    1 - \left(
        1 - \frac{1}{|\mathcal X||\mathcal T|}
    \right)^{-|\mathcal X||\mathcal T| \cdot \frac{-|\mathcal D|}{|\mathcal X||\mathcal T|}}
\right]^2\right\}^{|\mathcal X|}\\
&\sim \left\{ 1 - \left[
    1 - \exp\left(
        -\frac{|\mathcal D|}{|\mathcal {X}||\mathcal {T}|}
    \right)
\right]^2\right\}^{|\mathcal X|}\\
&= \left[ 2 \exp\left(
    -\frac{|\mathcal D|}{|\mathcal {X}||\mathcal {T}|}
\right) - \exp\left(
    -\frac{2|\mathcal D|}{|\mathcal {X}||\mathcal {T}|}
\right)\right]^{|\mathcal X|}\\
&= \exp\left(
    -\frac{|\mathcal D|}{|\mathcal {T}|}
\right)\left[
    2 - \exp\left(
        -\frac{|\mathcal D|}{|\mathcal X||\mathcal {T}|}
    \right)
\right]^{|\mathcal X|},
\end{align*}
where the third line uses $(1+\nicefrac{1}{n})^n\rightarrow e$ when $n\rightarrow\infty$.

\end{proof}

We next provide a lemma to estimate the probability that a random graph is connected.

\begin{lemma}
    \label{lem:random_graph}
    \citep{gilbert1959random} Suppose a graph $G$ is constructed
from a set of $N$ nodes in which each one of the $\nicefrac{N(N-1)}{2}$ possible links is present with probability $p$ independently. The probability that $G$ is connected can be estimated as:
\begin{align*}
    P(\text{connected}\mid G) \sim 1 - N (1 - p)^{N-1}.
\end{align*}

\end{lemma}

Applying Theorem \ref{thm:identifiability_prob}, Lemma \ref{lem:link_prob} and Lemma \ref{lem:random_graph}, we obtain:
\begin{align*}
P(\text{identifiability}\mid \mathcal D)&\sim 1 - |\mathcal T| \left\{
    \exp\left(
        -\frac{|\mathcal D|}{|\mathcal {T}|}
    \right)\left[
        2 - \exp\left(
            -\frac{|\mathcal D|}{|\mathcal X||\mathcal {T}|}
        \right)
    \right]^{|\mathcal X|}
\right\}^{|\mathcal T| - 1} \\
&= 1 - |\mathcal T| \exp\left[
    -|\mathcal D|\left(1 - \frac{1}{|\mathcal {T}|} \right)
\right] \left[
    2 - \exp\left(
        -\frac{|\mathcal D|}{|\mathcal X||\mathcal {T}|}
    \right)
\right]^{|\mathcal X|(|\mathcal T| - 1)}\\
&\sim 1 - |\mathcal T| \exp\left(
    -|\mathcal D|
\right) \left[
    2 - \exp\left(
        -\frac{|\mathcal D|}{|\mathcal X||\mathcal {T}|}
    \right)
\right]^{|\mathcal X||\mathcal T|}\\
&=1 - |\mathcal T| \exp\left(
    -|\mathcal D| + |\mathcal X||\mathcal {T}|\log \left[
        2 - \exp\left(
            -\frac{|\mathcal D|}{|\mathcal X||\mathcal {T}|}
        \right)
    \right]
\right)
\end{align*}
where the third line uses $\nicefrac{1}{|\mathcal T|} \rightarrow 0$ and $|\mathcal{T}| - 1 \rightarrow |\mathcal{T}|$ when $|\mathcal T|$ is large enough.

\subsection{Proof for Proposition \ref{thm:variance}}

\label{sec:app_proof_3}

Note that $Nr'(\vx)o'(\vt)$ follows a binomial distribution, \ie, 
\begin{align*}
    Nr'(\vx)o'(\vt)\sim B(N, r(\vx)o(\vt)),
\end{align*}
which implies:
\begin{align*}
    \mathbb E[Nr'(\vx)o'(\vt)]=Nr(\vx)o(\vt),\quad \mathbb V[Nr'(\vx)o'(\vt)] = N r(\vx)o(\vt)[1 - r(\vx)o(\vt)].
\end{align*}
Denote $g(\vt)=\nicefrac{o'(\vt)}{o(\vt)}$, then we have:
\begin{align*}
    \mathbb E[g(\vt)\mid r'(\vx)] &= \frac{Nr(\vx)o(\vt)}{Nr'(\vx)o(\vt)}=\frac{r(\vx)}{r'(\vx)},\\
    \mathbb V[g(\vt)\mid r'(\vx)]&=\frac{N r(\vx)o(\vt)[1 - r(\vx)o(\vt)]}{[Nr'(\vx)o(\vt)]^2}=\frac{r(\vx)(1-r(\vx)o(\vt))}{N r'(\vx)^2 o(\vt)}.
\end{align*}
Since $c(\vx, \vt)$ are sampled i.i.d., $g(\vt)$ is independent of $g(\vt')$ conditioned on $r'(\vx)$. Therefore,
\begin{align*}
    \mathbb E[g(\vt_1) - g(\vt_2)\mid r'(\vx)] &= \frac{r(\vx)}{r'(\vx)} - \frac{r(\vx)}{r'(\vx)} = 0,\\
    \mathbb V[g(\vt_1) - g(\vt_2)\mid r'(\vx)] &= \mathbb V[g(\vt_1)\mid r'(\vx)] + \mathbb V[g(\vt_2)\mid r'(\vx)]\\
    &=\frac{r(\vx)(1-r(\vx)o(\vt_1))}{N r'(\vx)^2 o(\vt_1)} + \frac{r(\vx)(1-r(\vx)o(\vt_2))}{N r'(\vx)^2 o(\vt_2)}\\
    &=\frac{r(\vx)^2}{N r'(\vx)^2} \left[
        \frac{1-r(\vx)o(\vt_1)}{r(\vx)o(\vt_1)}+
        \frac{1-r(\vx)o(\vt_2)}{r(\vx)o(\vt_2)}
    \right]\\
    &=\frac{1}{N R} \left[
        \frac{1}{r(\vx)o(\vt_1)}+
        \frac{1}{r(\vx)o(\vt_2)} - 2
    \right].
\end{align*}




\subsection{Proof for Proposition \ref{thm:merging_bound}}

\label{sec:app_proof_4}

We first separate the dataset $\mathcal D$ into two parts: $\mathcal{D}_1$ (corresponding to $G_1$) and $\mathcal{D}_2$ (corresponding to $G_2$), formally,
\begin{align*}
    \mathcal{D}_1&=\{(\vx, \vt)\in D\mid \vt\in V_1\},\\
    \mathcal{D}_2&=\{(\vx, \vt)\in D\mid \vt\in V_2\}.
\end{align*}

According to Theorem \ref{thm:graph}, the relevance model $r(\vx)$ ($\vx\in \{\vx_i\mid (\vx_i, \vt_i)\in \mathcal{D}_1 \}$) is identifiable on the dataset $\mathcal{D}_1$, and the relevance model $r(\vx)$ ($\vx\in \{\vx_i\mid (\vx_i, \vt_i)\in \mathcal{D}_2 \}$) is identifiable on the dataset $\mathcal{D}_2$. That is,
\begin{align}
    \nonumber
    \frac{r'(\vx_a)}{r(\vx_a)} = \frac{r'(\vx_b)}{r(\vx_b)}, \quad\forall \vx_a, \vx_b \in \{\vx_i\mid (\vx_i, \vt_i)\in \mathcal{D}_1 \},\\
    \label{eq:thm4_identifiable_two}
    \frac{r'(\vx_c)}{r(\vx_c)} = \frac{r'(\vx_d)}{r(\vx_d)}, \quad\forall \vx_c, \vx_d \in \{\vx_i\mid (\vx_i, \vt_i)\in \mathcal{D}_2 \}.
\end{align}

Since we have assumed that $\vx_1$ and $\vx_2$ appear in $\mathcal D$, we can find $\vt_1$ and $\vt_2$ such that $(\vx_1, \vt_1)\in\mathcal D$ and $(\vx_2, \vt_2)\in\mathcal D$.

(1) If $\vt_1\in V_1 \land \vt_2\in V_1$, or $\vt_1\in V_2 \land \vt_2\in V_2$, then according to \Eqref{eq:thm4_identifiable_two}, 
\begin{align}
    \label{eq:thm4_1}
    \left|\frac{r'(\vx_1)}{r(\vx_1)} - \frac{r'(\vx_2)}{r(\vx_2)}\right| = 0.
\end{align}

(2) Otherwise, without loss of generality we suppose $\vt_1\in V_1 \land \vt_2\in V_2$. For $t'$ and $t''$, we can find $\vx'$ and $\vx''$ such that $(\vx', \vt')\in\mathcal D_1$ and $(\vx'', \vt'')\in\mathcal D_2$. According to \Eqref{eq:thm4_identifiable_two}, 
\begin{align*}
    \frac{r'(\vx_1)}{r(\vx_1)} = \frac{r'(\vx')}{r(\vx')},\quad \frac{r'(\vx_2)}{r(\vx_2)} = \frac{r'(\vx'')}{r(\vx'')}.
\end{align*}

Since
\begin{align*}
    \frac{r'(\vx')}{r(\vx')} = \frac{o(\vt')}{o'(\vt')},\quad \frac{r'(\vx'')}{r(\vx'')} = \frac{o(\vt'')}{o'(\vt'')},
\end{align*}
we have
\begin{align}
    \label{eq:thm4_2}
    \left|
        \frac{r'(\vx_1)}{r(\vx_1)} - \frac{r'(\vx_2)}{r(\vx_2)}
    \right|=\left|
        \frac{o(\vt')}{o'(\vt')} - \frac{o(\vt'')}{o'(\vt'')}
    \right|=\left|
        \frac{o(\vt') - o(\vt'')}{o'(\vt')}
    \right|,
\end{align}
where we use the fact that $o'(\vt')=o'(\vt'')$. 

Combining \Eqref{eq:thm4_1} and \Eqref{eq:thm4_2} we obtain the desired result.

\subsection{Further theoretical analysis on node merging}\label{sec:error_bound_merging}

In this section, we provide further theoretical analysis on node merging, including the consistency guarantee and the error bound. We first formally introduce the Minimum Spanning Tree (MST) construction process of node merging, laying the groundwork for the subsequent analysis. Suppose an IG consists of $K$ connected components $\{G_i=(V_i, E_i)\}_{i=1}^K$. A node merging algorithm merges $K-1$ pairs of nodes $\mathcal M=\{(\vt_{a_i}, \vt_{b_i})\}_{i=1}^{K-1}$, forcing $o'(\vt_{a_i})=o'(\vt_{b_i})$. Based on it, we construct the following weighted connected graph $\mathcal G=(\mathcal V, \mathcal E)$ built on $K$ components:
\begin{align*}
    \mathcal V&=\{G_1,G_2,\cdots,G_K\},\\
    \mathcal E&=\{(G_i, G_j, w_{i,j})\mid \exists (\vt_i, \vt_j)\in \mathcal M, \text{s.t., } \vt_i\in V_i , \vt_j\in V_j , w_{i,j} := |\nicefrac{o(\vt_i)-o(\vt_j)}{o'(\vt_i)}| \},
\end{align*}
where $(G_i, G_j, w)$ denotes an edge connecting $(G_i, G_j)$ with a weight $w$. Notably, each merging pair in $\mathcal M$ corresponds to an edge in $\mathcal E$ exactly.

\subsubsection{Consistency of node merging}

Node merging enforces a constraint on function $o'$ to have identical values at certain bias factors. A concern arises that this additional constraint might conflict with the conditions for identifiability (\ie, $c \cdot r = c' \cdot r'$). Fortunately, such conflicts are absent when the graph $\mathcal G$ is acyclic, ensuring the node merging approach is \textbf{consistent}, as detailed in the following proposition.

\begin{proposition}[Consistency of node merging]
    Given an acyclic graph $\mathcal G$ with merging pairs $\mathcal M={(\vt_{a_i}, \vt_{b_i})}_{i=1}^{K-1}$, the true relevance $r(\cdot)$, and the true observation $o(\cdot)$, we can find functions $r'$ and $o'$ such that they satisfy:

    Condition 1: Unbiased click probability estimation: $r(\vx) \cdot o(\vt) = r'(\vx) \cdot o'(\vt), \quad \forall (\vx, \vt)\in\mathcal D$;

    Condition 2: Compliance with node merging constraints: $o'(\vt_{a_i})=o'(\vt_{b_i}), \quad \forall (\vt_{a_i}, \vt_{b_i}) \in \mathcal M$. 
\end{proposition}

\begin{proof}
    The proof employs induction.
    
    \textbf{Base step}: At $K=1$, the IG is connected, free from node merging constraints. Setting $r'(\cdot)=C r(\cdot)$ and $o'(\cdot)=\nicefrac{o(\cdot)}{C}$ for any positive constant $C$ satisfies unbiased click probability estimation.
    
    \textbf{Inductive step}: Assuming conditions 1 and 2 hold for $K=n$ ($n\geq 1$), we examine $K=n+1$. Given $\mathcal G$ as an acyclic graph, we can find a $G_i=(V_i, E_i)\in \mathcal V$ with a single edge connecting it. Let the corresponding merging pair be $(\vt_a, \vt_b)$, where $\vt_a\in V_i$, and $\vt_b$ belongs to another node in $G\; \textbackslash\; G_i$. The two conditions are met in $G\; \textbackslash\; G_i$ (with $n$ nodes) by induction assumption. We need only find $r'$ and $o'$ for $G_i$, as follows:
    \begin{align}
        \label{eq:construct_consistency}
        o'(\vt) = o(\vt)\frac{o'(\vt_b)}{o(\vt_a)}, \quad r'(\vx) = r(\vx)\frac{o(\vt_a)}{o'(\vt_b)}, \quad \forall (\vx, \vt)\in \mathcal D \land \vt \in V_i,
    \end{align}
    with $o'(\vt_b)$'s value fixed in $G\; \textbackslash\; G_i$, meeting both conditions by induction. It is easy to verify that \Eqref{eq:construct_consistency} meets both conditions as well.
    
\end{proof}

Note that in cases where $\mathcal G$ is not a tree and contains cycles, the inductive step fails at the step that finding a node with a single edge connecting it, resulting in inconsistency.

\subsubsection{Error bound for node merging}

Based on the above notation, we derive the following error bound for node merging.

\begin{proposition}[Error bound of node merging]
    \begin{align*}
        r(\vx)\cdot o(\vt) = r'(\vx)\cdot o'(\vt) \quad \Longrightarrow \quad \left|\frac{r'(\vx_1)}{r(\vx_1)} - \frac{r'(\vx_2)}{r(\vx_2)}\right| \leq D(\mathcal G),\quad \forall \vx_1, \vx_2\in\mathcal X,
    \end{align*}
    where $D(\mathcal G)$ is the diameter of $\mathcal G$, \ie, the maximal distance between nodes in $\mathcal G$.
\end{proposition}
\begin{proof}
In our proof, we employ Proposition \ref{thm:merging_bound} and the triangle inequality on the paths within $\mathcal G$. Initially, define $\vx(G_i) = \{\vx \mid \exists \vt \in V_i, \text{s.t. } (\vx, \vt) \in \mathcal D\}$ to represent all features associated with the bias factors in component $G_i$. It is crucial to note that a feature cannot belong to multiple components, as this would imply connectivity between these components through the said feature.

    For any two components $G_s$ and $G_t$, we find a path in $\mathcal G$ connecting them, \ie, $G_{a_1}\rightarrow G_{a_2}\rightarrow\cdots\rightarrow G_{a_n}$ with $a_1=s$ and $a_n=t$. We have:
    \begin{align*}
        \forall \vx_1\in \vx(G_{a_1}), \vx_2\in \vx(G_{a_2}), &\cdots, \vx_n\in \vx(G_{a_n}),\\
        \left|\frac{r'(\vx_1)}{r(\vx_1)} - \frac{r'(\vx_n)}{r(\vx_n)}\right| &= \left|
        \left(
            \frac{r'(\vx_1)}{r(\vx_1)} - \frac{r'(\vx_2)}{r(\vx_2)}
        \right) + \left(
            \frac{r'(\vx_2)}{r(\vx_2)} - \frac{r'(\vx_3)}{r(\vx_3)}
        \right) + \cdots + \left(
            \frac{r'(\vx_{n-1})}{r(\vx_{n-1})} - \frac{r'(\vx_n)}{r(\vx_n)}
        \right)
        \right|\\
        &\leq \left|
            \frac{r'(\vx_1)}{r(\vx_1)} - \frac{r'(\vx_2)}{r(\vx_2)}
        \right| + \left|
            \frac{r'(\vx_2)}{r(\vx_2)} - \frac{r'(\vx_3)}{r(\vx_3)}
        \right| + \cdots + \left|
            \frac{r'(\vx_{n-1})}{r(\vx_{n-1})} - \frac{r'(\vx_n)}{r(\vx_n)}
        \right|\\
        &\leq w_{a_1,a_2} + w_{a_2,a_3} + \cdots + w_{a_{n-1},a_n}.
    \end{align*}

    Here, the first inequality is derived using the triangle inequality, while the second follows from the definition of $\mathcal G$ and Proposition \ref{thm:merging_bound}.    Consequently, the error between $\vx_1$ and $\vx_n$ $\left|\nicefrac{r'(\vx_1)}{r(\vx_1)} - \nicefrac{r'(\vx_n)}{r(\vx_n)}\right|$ is limited by the cumulative weights of the path connecting $G_s$ and $G_t$. Given the arbitrary nature of this selection, the error between any two features is constrained by the diameter of $\mathcal G$.
\end{proof}

\section{Algorithms}

\subsection{Identifiability check}

\label{sec:app_algo_idcheck}

Based on Theorem \ref{thm:graph}, we illustrate the identifiability check in Algorithm \ref{alg:identifiability_check}. In lines 1-8, we construct a mapping, from feature to the bias factor sets that ever appear together with it. In lines 9-15, we connect the bias factor set for each feature to a complete graph, since they are all related to the same feature and thus connect.

\SetKwComment{Comment}{$\triangleright$\ }{}

\begin{algorithm}

\label{alg:identifiability_check}
\caption{Identifiability check}
\KwIn{Dataset $\mathcal D=\{ (\vx_i, \vt_i) \}_{i=1}^{|\mathcal D|}$}
\KwOut{Whether a relevance model trained on $\mathcal D$ is identifiable}
$S \gets \text{Dictionary()}$ \Comment*[r]{Initialize $S$ with an empty dictionary}
\Comment{Construct a mapping: feature $\rightarrow$ bias factors list}
\For{$i=1$ \KwTo $|\mathcal D|$}{ 
    \uIf{$\vx_i \notin S$}{
        $S[\vx_i] \gets \{\vt_i\} $\;
    }
    \Else{
    $S[\vx_i] \gets S[\vx_i] \cup \{\vt_i\} $\;
    }
}
$V \gets \{v_1, v_2, \cdots, v_{|\mathcal{T}|}\}$\;
$E \gets \varnothing$\;
\Comment{Construct identifiability graph (IG) based on $S$}
\For{$\vx\in S$}{
    \For{$\vt_1, \vt_2 \in S[\vx] \times S[\vx]$}{
        $E \gets E \cup \{ (v_{\vt_1}, v_{\vt_2}) \}$\;
    }
}
\uIf{$G=(V, E)$ is connected}{
    \Return{true}
}
\Else{
    \Return{false}
}
\end{algorithm}

\subsection{Full algorithm for node intervention}
\label{sec:app_algo_ni}

We illustrate the full algorithm for node intervention (\sref{sec:node_intervention}) in Algorithm \ref{alg:node_intervention}. Here we use Prim's algorithm to find the MST. In lines 1-2, we construct the IG and find $K$ connected components. In line 3, we initialize the intervention set. In line 4, we initialize the found node set to the first connected components for running Prim's algorithm. In line 7, we traverse the components in the unfound set $U-U'$ (denoted by $G_i$) and in the found set $U'$ (denoted by $G_j$), and compute the intervention cost and the best intervention pair between $G_i$ and $G_j$ in line 8. If the cost is the best, we record the cost, intervention pair, and the target component in lines 10-12. Finally, we add the best intervention pair we found in line 15 and update the found set in line 16 for Prim's algorithm. 

\begin{algorithm}

\label{alg:node_intervention}
\caption{Node intervention}
\KwIn{Dataset $\mathcal D=\{ (\vx_i, \vt_i) \}_{i=1}^{|\mathcal D|}$}
\KwOut{Intervention set $E'$}
Construct the IG $G=(V, E)$ on $\mathcal D$ using Algorithm \ref{alg:identifiability_check}\;
$U\gets \{G_1=(V_1, E_1), \cdots, G_K=(V_K, E_K)\}$ denoting $K$ connected components for $G$\;
$E'\gets \{\}$ \Comment*[r]{Initialize intervention set}
$U'\gets \{G_1\}$ \Comment*[r]{Initialize found nodes for Prim's algorithm}
\Comment{Construct an MST using Prim's algorithm}
\While(){$|U'|\neq K$}{
    $c_{\min} = + \infty$\;
    \For{$G_i\in U - U', G_j\in U'$}{
        Compute $\vx^{(i, j)}, \vt^{(i, j)}$ and the intervention cost $c$ based on \Eqref{eq:node_intervention_cost} - \Eqref{eq:node_intervention_d}\;
        \If{$c< c_{\min}$}{
            $c_{\min}\gets c$\;
            $\vx^*, \vt^*\gets \vx^{(i, j)}, \vt^{(i, j)}$\;
            $G^*\gets G_j$\;
        }
    }
    $E'\gets E'\cup \{(\vx^*, \vt^*)\}$ \Comment*[r]{Add the best intervention pair}
    $U'\gets U'\cup \{G^*\}$ \Comment*[r]{Update found nodes using Prim's algorithm}
}
\Return{$E'$}
\end{algorithm}

\subsection{Full algorithm for node merging}
\label{sec:app_algo_nm}

We illustrate the full algorithm for node merging (\sref{sec:node_merging}) in Algorithm \ref{alg:node_merging}. 

\begin{algorithm}

\label{alg:node_merging}
\caption{Node merging}
\KwIn{Dataset $\mathcal D=\{ (\vx_i, \vt_i) \}_{i=1}^{|\mathcal D|}$}
\KwOut{Merging set $E'$}
Construct the IG $G=(V, E)$ on $\mathcal D$ using Algorithm \ref{alg:identifiability_check}\;
$U\gets \{G_1=(V_1, E_1), \cdots, G_K=(V_K, E_K)\}$ denoting $K$ connected components for $G$\;
$E'\gets \{\}$ \Comment*[r]{Initialize merging set}
$U'\gets \{G_1\}$ \Comment*[r]{Initialize found nodes for Prim's algorithm}
\Comment{Construct an MST using Prim's algorithm}
\While{$|U'|\neq K$}{
    $c_{\min} = + \infty$\;
    \For{$G_i\in U - U', G_j\in U'$}{
        Compute $\vt^{*}_i, \vt^{*}_j$ and the merging cost $c$ on \Eqref{eq:node_merge_cost} - \Eqref{eq:node_merge_t12}\;
        \If{$c< c_{\min}$}{
            $c_{\min}\gets c$\;
            $\vt^*_{A}, \vt^*_{B}\gets \vt^{*}_i, \vt^{*}_j$\;
            $G^*\gets G_j$\;
        }
    }
    $E'\gets E'\cup \{(\vt^*_{A}, \vt^*_{B})\}$ \Comment*[r]{Add the best merging pair}
    $U'\gets U'\cup \{G^*\}$ \Comment*[r]{Update found nodes using Prim's algorithm}
}
\Return{$E'$}
\end{algorithm}

Similar to node intervention, here we also use Prim's algorithm to find the MST. In lines 1-2, we construct the IG and find $K$ connected components. In line 3, we initialize the merging set. In line 4, we initialize the found node set to the first connected components for running Prim's algorithm. In line 7, we traverse the components in the unfound set $U-U'$ (denoted by $G_i$) and in the found set $U'$ (denoted by $G_j$) and compute the merging cost and the best intervention pair between $G_i$ and $G_j$ in line 8. If the cost is the best, we record the cost, merging the pair and the target component in lines 10-12. Finally, we add the best merging pair we found in line 15 and update the found set in line 16 for Prim's algorithm.

\newpage

\begin{figure*}[htbp]
    \centering
    \subfigure[$K=1$]{
        \includegraphics[width=0.1\textwidth]{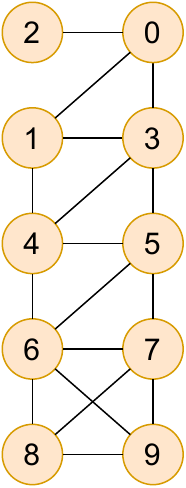}
    }
    \hspace{2.5em}
    \subfigure[$K=2$]{
        \includegraphics[width=0.1\textwidth]{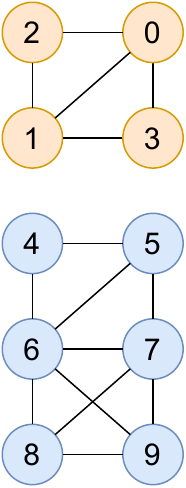}
    }
    \hspace{2.5em}
    \subfigure[$K=3$]{
        \includegraphics[width=0.1\textwidth]{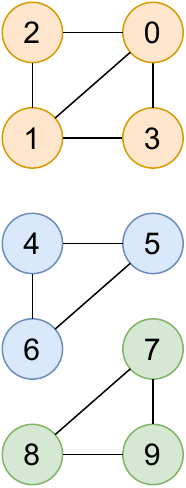}
    }
    \hspace{2.5em}
    \subfigure[$K=4$]{
        \includegraphics[width=0.1\textwidth]{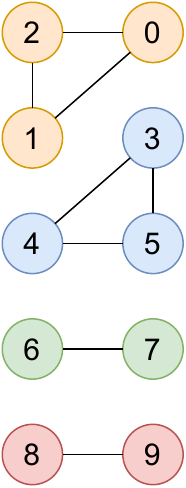}
    }
    \caption{IGs of the simulated datasets. Numbers in the nodes denote the position index (starting from 0 to 9). }
    \label{fig:experiment_simulation_ig}
\end{figure*}

\begin{table}[htbp]
    \centering
    \begin{minipage}{0.45\linewidth}
    
    \centering
    \caption{Dataset statistics}
    \label{tab:data_statistics}
    \small
    \begin{tabular}{lll}
    \toprule
              & Yahoo!  & Istella-S \\ \midrule
    \# Queries   & 28,719  & 32,968    \\
    \# Documents & 700,153 & 3,406,167 \\
    \# Features  & 700     & 220       \\
    \# Relevance levels  & 5     & 5       \\ \bottomrule
    \end{tabular}
    \end{minipage}
    \hfill
    \begin{minipage}{0.45\linewidth}
    \caption{Identifiability graph statistics}
    \centering
    \small
    \label{tab:ig_real}
    \begin{tabular}{lll}
    \toprule
              & Yahoo!  & Istella-S \\ \midrule
    \# Nodes   & 48,894  & 48,919    \\
    \# Edges & 59,811,898 & 5,972 \\
    \# Connected components (CC)  &  28,958    & 47,976       \\
    \# Nodes in the Top 1 CC  &  15,684    & 101       \\ 
    \# Nodes in the Top 2 CC  &  98  & 7       \\ 
    \# Nodes in the Top 3 CC  &  96  & 5       \\ 
    \bottomrule
    \end{tabular}
    \end{minipage}
\end{table}

\section{Experiment details}

\label{sec:app_experiments}

\subsection{Datasets}

\label{sec:app_datasets}

In this section, we show details about the datasets we used in this work, including the simulated datasets and semi-synthetic datasets. The datasets can be downloaded from \url{https://webscope.sandbox.yahoo.com/} (Yahoo!), \url{http://quickrank.isti.cnr.it/istella-dataset/} (Istella-S) and \url{http://www.thuir.cn/tiangong-st/} (TianGong-ST).

\paragraph{Further details of the simulated datasets} 

For a fair comparison, all simulated datasets comprised 10,000 one-hot encoded documents and 1,150 queries with randomly sampled 5-level relevance, and each query contains 10 documents. Figure \ref{fig:experiment_simulation_ig} demonstrates the IGs we used for simulating datasets, where the number of connected components is $K=1,2,3,4$ respectively. 

\paragraph{Further details of the semi-synthetic datasets}

We followed the given data split of training, validation, and testing. To generate initial ranking lists for click simulation, we followed the standard process \citep{joachims2017unbiased,ai2018unbiased,chen2021adapting,chen2022lbd} to train a Ranking SVM model \citep{joachims2006training} with 1\% of the training data with relevance labels, and sort the documents. We used the ULTRA framework \citep{ai2018ultra,ai2021unbiased} to pre-process datasets. Table \ref{tab:data_statistics} shows the characteristics of the two datasets we used. 

Since the IGs of the two datasets we used are too large to visualize, we show several graph characteristics about them in Table \ref{tab:ig_real} where the number of context types is 5,000.

\subsection{Click simulation}

\label{sec:app_click_simulation}

\paragraph{Position-based model}
We sampled clicks according to the examination hypothesis (\Eqref{eq:examination_hypothesis}) for fully simulated datasets. Following the steps proposed by \citet{chapelle2009expected}, we set the relevance probability to be:
\begin{align}
    r(\vx) = \epsilon + (1-\epsilon)\frac{2^{y_{\vx}}-1}{2^{y_{\max}}-1},
\end{align}
where $y_{\vx}$ is the relevance level of $\vx$, and $y_{\max}=5$ in our case. $\epsilon$ is the click noise level and we set $\epsilon=0.1$ as the default setting. For the observation part, following \citet{ai2021unbiased} we adopted the position-based examination probability $o(p)$ for each position $p$ by eye-tracking studies \citep{joachims2005accurately}. 

\paragraph{Contextual position-based model}
For simulating contextual bias on the semi-synthetic dataset, following \citet{fang2019intervention}, we assigned each context id $t$ with a vector $\mX_{t}\in\mathbb R^{10}$ where each element is drawn from $\mathcal N(0, 0.35)$. We followed the same formula as a position-based model for click simulation, while the observation probability takes the following formula:
\begin{align*}
    o(t, p) = o(p)^{\max\{\vw^\top \mX_{t}+1,0\}},
\end{align*}
where $o(p)$ is the position-based examination probability used in the fully synthetic experiment. $\vw$ is fixed to a 10-dimensional vector uniformly drawn from $[-1, 1]$.

\subsection{Training details} 

\label{sec:app_training}

\paragraph{Implementation details of baselines}

\textbf{\textit{DLA}} \citep{ai2018ultra,vardasbi2020cascade,chen2021adapting,chen2023multi} uses the following formula to learn the relevance model $r'$ and observation model $o'$ dually:
\begin{align*}
    r'_{k}(\vx) &\gets \arg\min_{r'(\vx)}\sum_{i=1}^{|\mathcal D|} 
    \mathbbm{1}_{\vx_i=\vx} \mathcal L(c_i, o'_{k-1}(\vt_i), r'(\vx)),\\
    o'_{k}(\vt) &\gets \arg\min_{o'(\vt)}\sum_{i=1}^{|\mathcal D|} \mathbbm{1}_{\vt_i=\vt}
    \mathcal L(c_i, o'(\vt), r'_{k-1}(\vx_i)),
\end{align*}
where $\mathcal D=\left\{(\vx_i, \vt_i, c_i)\right\}_{i=1}^{|\mathcal D|}$ is the dataset, $r'_{k}(\vx)$ and $o'_{k}(\vt)$ are the $k$-th step model output ($1\leq k\leq T$). The initial values for $o'_{0}$ and $r'_{0}$ were randomly initialized from a uniform distribution within the range of $[0, 1]$. After each step, we applied a clipping operation to constrain the outputs within the interval $[0, 1]$. $\mathcal L$ is the training objective function, and we implement it with MSE loss:
\begin{align*}
    \mathcal L(c_i, o_i', r_i')=(c_i-o_i'r_i')^2.
\end{align*}

\textbf{\textit{Regression-EM}} \citep{wang2018position,sarvi2023impact} uses an iterative process similar to DLA, while the relevance model $r'$ and observation model $o'$ are learned as follows:
\begin{align*}
    r'_{k}(\vx) \gets& \frac{\sum_{i=1}^{|\mathcal D|} \mathbbm{1}_{\vx_i=\vx}\left\{ c_i+(1-c_i)\frac{[1-o'_{k-1}(\vt_i)]r'_{k-1}(\vx_i)}{1-o'_{k-1}(\vt_i)r'_{k-1}(\vx_i)}\right\}}{\sum_{i=1}^{|\mathcal D|} \mathbbm{1}_{\vx_i=\vx}},\\
    o'_{k}(\vt) \gets& \frac{\sum_{i=1}^{|\mathcal D|} \mathbbm{1}_{\vt_i=\vt}\left\{ c_i+(1-c_i)\frac{[1-r'_{k-1}(\vx_i)]o'_{k-1}(\vt_i)}{1-o'_{k-1}(\vt_i)r'_{k-1}(\vx_i)}\right\}}{\sum_{i=1}^{|\mathcal D|} \mathbbm{1}_{\vt_i=\vt}}.\\
\end{align*}

\textbf{\textit{Two-Tower}} \citep{guo2019pal,cacm,chen2022lbd} treats $r'$ and $o'$ as two towers and facilitate the multiplication of the outputs of them close to clicks. They use the binary cross entropy loss to train the models by gradient descent, formulated as:
\begin{align*}
    \mathcal L = \sum_{i=1}^{|\mathcal D|}-c_i\log [r'(\vx_i)o'(\vt_i)] - (1-c_i) \log [1-r'(\vx_i)o'(\vt_i)],
\end{align*}
where $r'$ and $o'$ are constrained to $[0, 1]$ by applying a sigmoid function.

\paragraph{Hyper-parameters}

We ran each experiment 10 times and reported the average values as well as the standard deviations. On the fully synthetic datasets, we implemented the ranking and observation models as embedding models and controlled $T=20,000$ to ensure the convergence. On the semi-synthetic datasets, we also implemented the ranking and observation models as embedding models by assigning a unique identifier based on ranking features to each document, which improves the model's ability to fit clicks during training. The number of epochs was $T=100$. After training, to generalize to the test set, we trained a LightGBM \citep{ke2017lightgbm} as a ranking model with the learned relevance embeddings of each feature. The total number of trees
was 500, the learning rate was 0.1, number of leaves for one tree was
255.

\paragraph{Implementation details of node merging}
For node merging, we used the position number as the bias feature on the fully synthetic dataset. On the semi-synthetic dataset, we formed the bias feature $\mX_{p, t}$ for each bias (consisting of position $p$ and context type $t$) as follows: we multiplied $p$ by 10 and added it to the end of the $10$-dimensional context vector $\mX_{t}$, to form an $11$-dimensional bias feature. This method increases the weight of the position, forcing node merging to give priority to merging different context types rather than positions. 

\paragraph{Implementation details of node intervention}
For node intervention on the fully synthetic dataset, we trained the ranking model and observation model using node merging and used their values to implement the cost function (\Eqref{eq:node_intervention_cost}).

\section{Further experiment results}

\begin{figure*}[t]
    \centering
    \subfigure[$K=2$]{
    \includegraphics[width=0.31\textwidth,trim={10 10 10 10},clip]{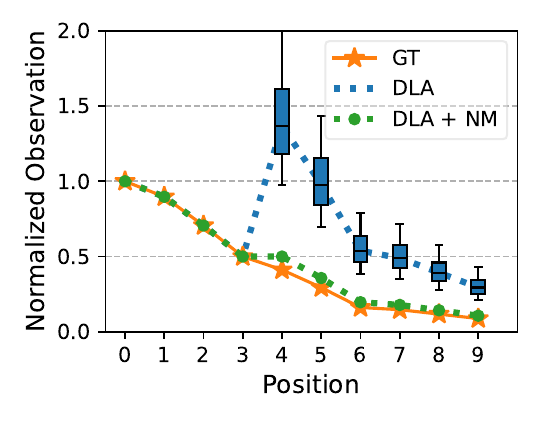}
        \label{fig:observation_k2}
    }
    \subfigure[$K=3$]{
        \includegraphics[width=0.31\textwidth,trim={10 10 10 10},clip]{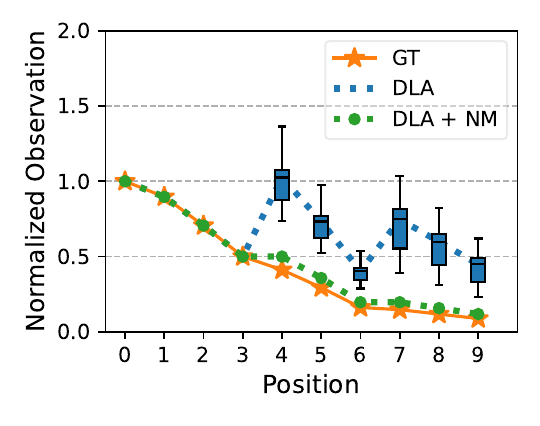}
        \label{fig:observation_k3}
    }
    \subfigure[$K=4$]{
    \includegraphics[width=0.31\textwidth,trim={10 10 10 10},clip]{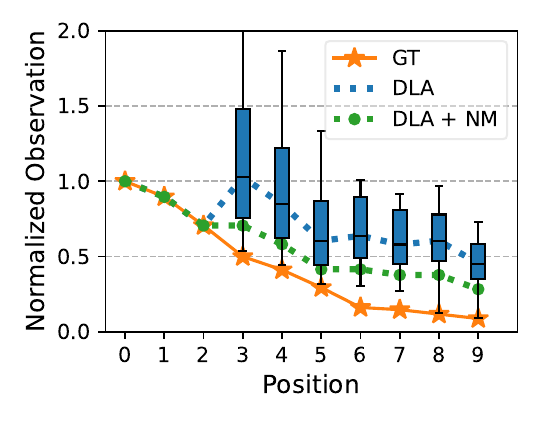}
        \label{fig:observation_k4}
    }
    \caption{Observation curves on each fully simulated dataset. GT = \underline{G}round \underline{T}ruth. NM = \underline{N}ode \underline{M}erging.}
    \label{fig:experiment_observation}
\end{figure*}

\subsection{Observation probability curves}\label{sec:observation}

Observation probability estimation, also known as propensity estimation \citep{agarwal2019estimating,fang2019intervention}, is a critical task in ULTR. Figure \ref{fig:experiment_observation} illustrates ground truth (GT) for observation probability for $K=2$, $K=3$, and $K=4$ datasets, alongside estimations by DLA and node merging (NM). We normalize the model's predictions by dividing the predicted probabilities at each position by the predicted probability of the first position. One can observe substantial volatility in DLA’s predictions under unidentifiability, as reflected in the boxplots. Node merging, by mandating shared observation probabilities for disconnected positions on the IG, achieves model identifiability and closer alignment with the GT observation curve. Notably, node merging’s accuracy for the $K=4$ dataset is inferior to that of $K=3$ and $K=2$, due to the merging of positions with larger observational disparities in $K=4$. This is attributed to the merging of positions with greater differences in observations in $K=4$ (positions 2 versus 3) since they are disconnected (visualized in Figure \ref{fig:experiment_simulation_ig}), as opposed to smaller differences in $K=3$ and $K=2$. Consequently, this results in a relatively poorer performance of node merging in recovering relevance for $K=4$ (as seen in Figure \ref{fig:dataset_components_method}). This phenomenon further corroborates Proposition \ref{thm:merging_bound} and illustrates the dependence of node merging’s performance on dataset characteristics.

\subsection{Further investigation on the identifiability of TianGong-ST}
\label{sec:app_tiangong}

On the TianGong-ST dataset, vertical types are represented in the format ``$v_1\text{\#}v_2$'', \eg, ``-1$\text{\#}$-1'' or ``30000701$\text{\#}$131''. We investigated the consequences of excluding specific bias factors, for example, disregarding either $v_1$ or $v_2$. A summary of the results is presented in Table \ref{tab:ig_tiangong}. Our findings reveal that when one of $v_1$ and $v_2$ is kept, the IG loses its connectivity, demonstrating the prevalence of unidentifiability issues in real-world scenarios. However, when only the position is retained, the IG regains connectivity, rendering it identifiable. We observed the reason is that some queries are repeated in the dataset, with variations in the order of related documents. This observation suggests that the search engine may have done some position intervention during deployment which enhances the IG's connectivity. Overall, it reveals again that introducing excessive bias factors leads to a higher probability of unidentifiability.

\begin{table}[htbp]
    \small
    \caption{Identifiability of TianGong-ST in different bias factor settings.}
    \centering
    \label{tab:ig_tiangong}
    \begin{tabular}{rrrr}
    \toprule
    Bias factor & \# Connected components & Identifiable? \\ \midrule
    ($v_1$, $v_2$, position) & 2,900 & $\times$ \\
    ($v_1$, position) & 1,106 & $\times$ \\
    ($v_2$, position) & 87 & $\times$ \\
    position only & 1 & $\checkmark$ \\
    \bottomrule
    \end{tabular}
\end{table}

\subsection{Impact of the sampling ratio and feature count on the identifiability probabilities}

\label{sec:app_sampling_ratio}

We sampled the datasets randomly according to a sampling ratio 20 times and calculated the frequency that the IG calculated on the sampled datasets is connected, when \textit{positions are the only bias factors}. From Figure \ref{fig:sample_ratio_yahoo} and \ref{fig:sample_ratio_istella}, one can find that although the IGs on both the entire datasets are connected, the subsets of the datasets are not. Specifically, the connectivity of IGs can be guaranteed only when the size of the subset comprises more than approximately 2\% (in the case of Yahoo!) and 50\% (in the case of Istella-S) of the respective dataset.

\begin{figure*}[t]
    \centering
    \subfigure[Yahoo!]{
    \includegraphics[width=0.28\textwidth,trim={10 10 10 10},clip]{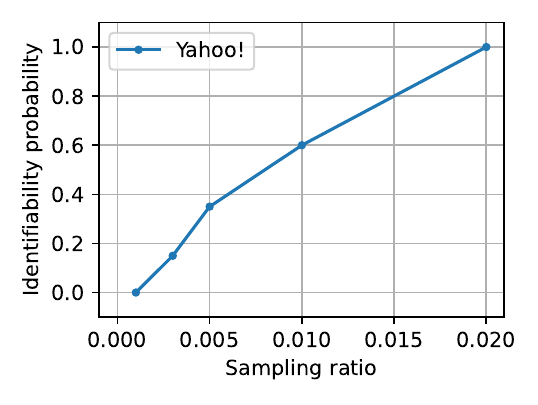}
        \label{fig:sample_ratio_yahoo}
    }
    \subfigure[Istella-S]{
        \includegraphics[width=0.28\textwidth,trim={10 10 10 10},clip]{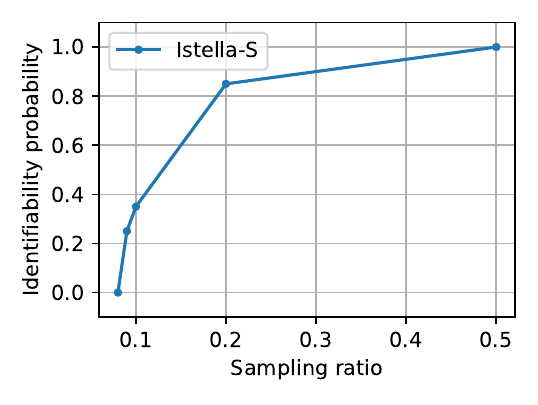}
        \label{fig:sample_ratio_istella}
    }\\
    \subfigure[Yahoo!]{
    \includegraphics[width=0.28\textwidth,trim={10 10 10 10},clip]{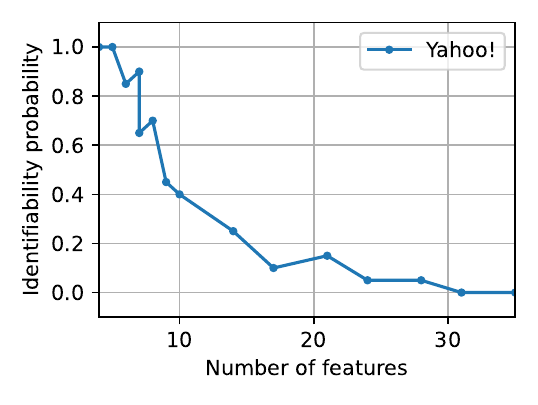}
        \label{fig:sample_feature_yahoo}
    }
    \subfigure[Istella-S]{
        \includegraphics[width=0.28\textwidth,trim={10 10 10 10},clip]{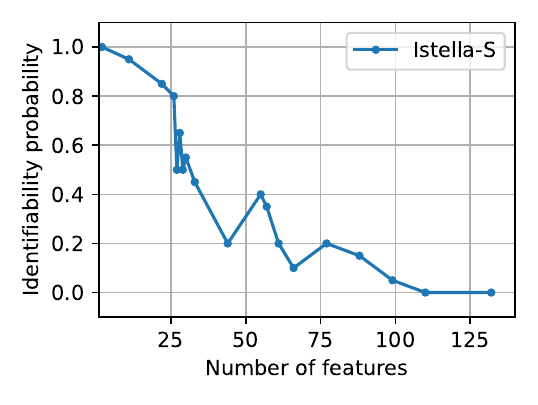}
        \label{fig:sample_feature_istella}
    }
    \caption{The influence of (a)(b) the sampling ratio and (c)(d) number of features of datasets on the identifiability probabilities, on Yahoo! and Istella-S, respectively. }
    \label{fig:experiment_semi_sample}
\end{figure*}

Subsequently, we explored the impact of feature count on identifiability. Due to dataset constraints, increasing the number of features was not feasible. Therefore, we adopted a strategy of progressively removing features from an unidentifiable dataset until the IG became connected. As both datasets are identifiable in the absence of context, we generated 500 contexts for Yahoo! and 50 for Istella-S, as depicted in Figures \ref{fig:sample_context_yahoo} and \ref{fig:sample_context_istella}. The results are displayed in Figure \ref{fig:sample_feature_yahoo} and \ref{fig:sample_feature_istella}. One can observe that the probability of identifiability diminishes with an increase in feature count. Notably, in the presence of a sufficient number of bias factors, a relatively small number of features (5-10) is sufficient to induce unidentifiability.

\subsection{Impact of initialization strategies}\label{app:initialization}

\begin{table*}[htbp]
  \centering
  \small
  \caption{Performance comparison on two datasets under CPBM bias with different initialization strategies on the observation model. We ran each experiment 10 times and reported the average results as well as the standard deviations. }
  \fontsize{8pt}{8pt}\selectfont

\begin{tabular}{clccccccc}
\toprule
\multirow{2}[3]{*}{\textbf{Dataset}} & \multicolumn{1}{c}{\multirow{2}[3]{*}{\textbf{Method}}} & \multicolumn{4}{c}{\textbf{Training}} &      & \multicolumn{2}{c}{\textbf{Test}} \\
\arrayrulecolor{gray!80}\cmidrule{3-6}\cmidrule{8-9}\arrayrulecolor{black}     &      & \textbf{MCC}$\:\uparrow$  & \textbf{nDCG@5}$\:\uparrow$ & \textbf{nDCG@10}$\:\uparrow$ & \textbf{Click MSE} &      & \textbf{nDCG@5}$\:\uparrow$ & \textbf{nDCG@10}$\:\uparrow$ \\
\midrule
\multirow{5}[2]{*}{Yahoo} & No debias & $0.765_{\pm.000}$ & $0.841_{\pm.000}$ & $0.915_{\pm.000}$ & $3.7\times 10^{-4}$ &      & $0.693_{\pm.002}$ & $0.741_{\pm.001}$ \\
     & DLA$^\mathcal I$  & $0.750_{\pm.000}$ & $0.844_{\pm.000}$ & $0.914_{\pm.000}$ & $1.7\times 10^{-5}$ &      & $0.693_{\pm.001}$ & $0.741_{\pm.001}$ \\
     & DLA$^\mathcal R$ & $0.619_{\pm.003}$ & $0.819_{\pm.001}$ & $0.892_{\pm.001}$ & $1.7\times 10^{-5}$ &      & $0.688_{\pm.001}$ & $0.737_{\pm.001}$ \\
     & DLA$^\mathcal I$ + Node merging & \boldmath{}\textbf{$0.771_{\pm.000}$}\unboldmath{} & $0.853_{\pm.000}$ & $0.920_{\pm.000}$ & $4.4\times 10^{-5}$ &      & $0.697_{\pm.001}$ & $0.745_{\pm.001}$ \\
     & DLA$^\mathcal R$ + Node merging & $0.744_{\pm.002}$ & \boldmath{}\textbf{$0.863_{\pm.001}$}\unboldmath{} & \boldmath{}\textbf{$0.921_{\pm.001}$}\unboldmath{} & $4.3\times 10^{-5}$ &      & \boldmath{}\textbf{$0.699_{\pm.001}$}\unboldmath{} & \boldmath{}\textbf{$0.746_{\pm.001}$}\unboldmath{} \\
\arrayrulecolor{gray!80}
\midrule
\arrayrulecolor{black}
\multirow{5}[2]{*}{Istella-S} & No debias & $0.764_{\pm.000}$ & $0.885_{\pm.000}$ & $0.941_{\pm.000}$ & $4.4\times 10^{-5}$ &      & $0.634_{\pm.001}$ & $0.682_{\pm.001}$ \\
     & DLA$^\mathcal I$  & $0.764_{\pm.000}$ & $0.886_{\pm.000}$ & $0.941_{\pm.000}$ & $1.1\times 10^{-6}$ &      & $0.633_{\pm.001}$ & $0.682_{\pm.001}$ \\
     & DLA$^\mathcal R$ & $0.748_{\pm.001}$ & $0.874_{\pm.001}$ & $0.931_{\pm.000}$ & $1.1\times 10^{-6}$ &      & \boldmath{}\textbf{$0.638_{\pm.002}$}\unboldmath{} & \boldmath{}\textbf{$0.686_{\pm.002}$}\unboldmath{} \\
     & DLA$^\mathcal I$ + Node merging & $0.772_{\pm.000}$ & $0.892_{\pm.000}$ & $0.944_{\pm.000}$ & $2.2\times 10^{-5}$ &      & $0.636_{\pm.001}$ & $0.684_{\pm.001}$ \\
     & DLA$^\mathcal R$ + Node merging & \boldmath{}\textbf{$0.782_{\pm.001}$}\unboldmath{} & \boldmath{}\textbf{$0.900_{\pm.001}$}\unboldmath{} & \boldmath{}\textbf{$0.947_{\pm.000}$}\unboldmath{} & $2.2\times 10^{-5}$ &      & \boldmath{}\textbf{$0.638_{\pm.001}$}\unboldmath{} & \boldmath{}\textbf{$0.686_{\pm.001}$}\unboldmath{} \\
\bottomrule
\end{tabular}%

  \label{tab:experiment_real_initialization}%
\end{table*}%

In unidentifiable scenarios, models can converge to various parameters that predict click probability correctly, but only a limited range truly represents correct relevance. Therefore, the choice of model initialization plays a crucial role in determining convergence quality. We examined two strategies for initializing the observation model: setting all bias factors' observation probabilities to 1.0 (denoted by $\mathcal I$) and employing random initialization within [0.0, 1.0] (denoted by $\mathcal R$). Table \ref{tab:experiment_real_initialization} presents the performance of DLA and node merging using different initialization strategies, which show that DLA's efficacy on unidentifiable datasets is highly sensitive to initial observation probabilities, with randomization often leading to subpar convergence. In contrast, for identifiable datasets processed with node merging, DLA exhibits more consistent results, demonstrating resilience to varying initialization approaches and supporting our theory of identifiability.

\end{document}